\documentclass[aps,prl,reprint]{revtex4-2}

\usepackage{amsmath, amssymb, amsthm, mathtools, bm, xfrac, nicefrac, dsfont}

\usepackage[dvipsnames]{xcolor}
\usepackage[most]{tcolorbox}

\usepackage{graphicx}
\usepackage{pgfplots} 
\pgfplotsset{compat=1.18} 
\usepackage{tikz}
\usetikzlibrary{shapes,arrows.meta,positioning,calc}

\usepackage[mathscr]{euscript}
\usepackage{dcolumn}
\usepackage{algpseudocode}
\allowdisplaybreaks[2]

\usepackage{booktabs}
\usepackage{multirow}
\usepackage{makecell}
\usepackage{adjustbox}
\usepackage{array}
\newcolumntype{C}[1]{>{\centering\arraybackslash}m{#1}} 
\usepackage{colortbl}
\usepackage{arydshln} 

\usepackage[linktocpage=true]{hyperref}
\hypersetup{
    unicode=false,
    pdftoolbar=true,
    pdfmenubar=true,
    pdffitwindow=false,
    pdfstartview={FitH},
    pdfnewwindow=true,
    colorlinks=true,
    linkcolor=Red,
    citecolor=ForestGreen,
    filecolor=Magenta,
    urlcolor=BlueViolet,
}

\newcommand{\Renyi}{{}R\'{e}nyi{ }}
\newcommand{\EA}{\mathrm{EA}}
\newcommand{\Tr}{\operatorname{Tr}}

\newcommand{\Err}{\operatorname{Err}}

\newcommand{\I}{\mathbf 1}

\def\X{\mathsf{X}}
\def\Y{\mathsf{Y}}
\def\A{\mathsf{A}}
\def\B{\mathsf{B}}

\def\R{\mathsf{R}}

\theoremstyle{plain}
\newtheorem{theorem}{Theorem}
 
\newtheorem{proposition}[theorem]{Proposition}

\theoremstyle{definition}

\theoremstyle{remark}

\theoremstyle{plain}
\newtheorem{app-theorem}{Theorem}[section]
\newtheorem{app-lemma}[app-theorem]{Lemma}
\newtheorem{app-proposition}[app-theorem]{Proposition}
\theoremstyle{remark}
\newtheorem{app-remark}[app-theorem]{Remark}

\begin{document}

\title{Superadditivity for Entanglement-Assisted Communication}

\author{Hao-Chung Cheng}
\affiliation{Department of Electrical Engineering and Graduate Institute of Communication Engineering, National Taiwan University, Taiwan}
\affiliation{Center for Quantum Science and Engineering,  National Taiwan University, Taiwan}
\affiliation{Physics/Mathematics Division, National Center for Theoretical Sciences, Taiwan}
\affiliation{Hon Hai (Foxconn) Quantum Computing Center, Taiwan}
\author{Mario Berta}
\affiliation{Institute for Quantum Information, RWTH Aachen University, Germany}


\date{\today}

\begin{abstract}
The entanglement-assisted capacity of a quantum channel admits an additive single-letter characterization, implying that joint encodings across channel uses cannot increase the ultimate communication rate. Here, we show that this additive picture does not extend to communication reliability. Specifically, we prove that the Petz–Rényi channel information can be strictly superadditive for every \(\alpha\in[\sfrac12,1)\), yielding a genuine multi-copy enhancement of the entanglement-assisted random-coding error exponent, even though the entanglement-assisted capacity remains additive. We establish this phenomenon analytically already for measurement channels, which are entanglement-breaking and have additive unassisted capacity. Remarkably, this strict superadditivity is witnessed by a separable, classically correlated two-copy channel-input marginal, demonstrating that no entanglement between the transmitted systems is required. Our results show that, although correlations across channel uses cannot increase the ultimate rate of entanglement-assisted communication, they can enhance its reliability.
\end{abstract}

\maketitle

\section{Introduction}

The fundamental limits of reliable information transmission over a noisy point-to-point channel are quantified by its \emph{channel capacity}. 
For a classical channel $\mathscr{W}$, the capacity is characterized by the mutual information of channel $I_1(\mathscr{W})$ \cite{Sha48}.
Moreover, Shannon proved that the capacity is additive under independent uses of channels, i.e.~$I_1(\mathscr{W}_1 \otimes \mathscr{W}_2) = I_1(\mathscr{W}_1) + I_1(\mathscr{W}_2)$, reducing the evaluation of this fundamental quantity to a computable single-letter formula, cast as a fixed-dimensional convex optimization.
Operationally, the additivity of $I_1(\mathscr{W})$ means that encoding classical data jointly across both channels yields no advantage for the maximum achievable rate.

Quantum mechanics, however, fundamentally departs from this paradigm.
For classical communication over a quantum channel $\mathscr{N}$, the Holevo information
of a channel \(\chi(\mathscr N)\)~\cite{Hol98,SW97} need not be additive:
there exist quantum channels \(\mathscr{N}_1\) and \(\mathscr{N}_2\) such that
$
    \chi(\mathscr{N}_1\otimes\mathscr{N}_2)
    >
    \chi(\mathscr{N}_1) + \chi(\mathscr{N}_2)
$
\cite{hastings2009superadditivity,aubrun2011hastings}.
Consequently, ensembles containing entangled states across independent
channel uses can outperform product-state encodings, and hence the classical
capacity generally requires regularization over arbitrarily many channel
uses.
An even more striking nonadditivity arises in quantum communication
\cite{lloyd1997capacity,shor2002quantum,devetak2005private}: two channels
with individually vanishing quantum capacities can have positive quantum
capacity when used jointly, a phenomenon known as
\emph{superactivation}~\cite{smith2008quantum}.
More generally, no fixed finite block length suffices to determine the
quantum capacity of all channels~\cite{cubitt2015unbounded}.
These results establish correlations across channel uses as an operational
resource, while explaining why quantum channel capacities often resist
single-letter characterization. 
For unassisted communication, however, this nonadditivity disappears
for entanglement-breaking channels: their Holevo information is strongly
additive, and hence their classical capacity is single-lettered~\cite{shor2002additivity}.

To restore the elegant phenomenon of additivity for general quantum channels, preshared entanglement between the sender and receiver emerges as an operational resolution \cite{BSS+99,BSS+02,Hol02}.
Bennett \textit{et al.} showed that allowing the sender and receiver unlimited preshared entanglement reduces the classical capacity of a quantum channel, $C_{\mathrm{EA}}(\mathscr{N})$, to an additive, single-letter optimization of the quantum mutual information. The corresponding entanglement-assisted quantum capacity satisfies $Q_{\EA}(\mathscr{N}) = \frac12 C_{\EA}(\mathscr{N})$ by teleportation \cite{BBC+93} and superdense coding \cite{BW92}.
Hence, the regularization and superadditivity that obstruct unassisted capacity formulas disappear in the entanglement-assisted setting, yielding a natural quantum analogue of Shannon’s coding theorem. More broadly, this result reveals that the nonadditive complexity of quantum communication is not an immutable property of the channel alone, but depends fundamentally on which correlations are available as operational resources.

\section{Quality of communication}

While channel capacity establishes the ultimate \emph{quantity} of reliably transmissible information, the operational performance of practical physical systems is equally governed by the \emph{quality} of that communication. This quality is characterized by the \emph{error exponent} for rates below capacity and the \emph{strong converse exponent} for rates above capacity. Together, these exponents dictate the exponential decay of the decoding error probability and success probability, respectively, at any fixed transmission rate.
For transmission rates above capacity, the strong converse exponent \cite{GW14, LY24b} is determined by a simple formula in terms of the channel's sandwiched \Renyi information $\widetilde{I}_{\alpha}(\mathscr{N})$ of order $\alpha>1$.
Because this sandwiched quantity is additive in this regime \cite{GW14}, the entanglement-assisted strong converse exponent circumvents intractable asymptotic limits, yielding a computable single-letter formula.
Furthermore, the recent proof extending this result to $\alpha \in [\sfrac12, 1)$ \cite{li2026completely}, together with the mutual-information point \(\alpha=1\), demonstrates additivity of the sandwiched--\Renyi channel information throughout its full data-processing range \(\alpha\geq\sfrac12\).
Taken together, these results show that correlations across multiple channel uses provide no advantage for the entanglement-assisted capacity or the above-capacity strong-converse exponent.

Recent progress establishes an exponential decay rate for the decoding error probability, expressed in terms of the channel's Petz--\Renyi information \(I_{\alpha}(\mathscr{N})\) of order \(\alpha\in[\sfrac12,1)\) \cite{preparation}.
Given that entanglement assistance resolves the nonadditivity of channel capacity, prior results on the additivity of \(I_{1}(\mathscr{N})\) and \(\widetilde{I}_{\alpha}(\mathscr{N})\) for \(\alpha\geq\sfrac12\) naturally suggest that \(I_{\alpha}(\mathscr{N})\), and thereby the quality of entanglement-assisted communication, might exhibit the same well-behaved additive structure when coding with rates at capacity and above.

However, when turning to the practically more relevant regime of transmission rates below capacity, we prove that the Petz--\Renyi channel information exhibits strict superadditivity for every $\alpha\in[1/2,1)$, yielding a genuine multi-copy improvement in communication reliability.
Interestingly, this failure of additivity already occurs for an entanglement-breaking measurement channel, a class for which the Holevo information---and hence the unassisted classical capacity---is additive \cite{shor2002additivity}. 
Hence, even a channel that outputs only classical data and whose output is necessarily separable from any retained reference can exhibit a collective entanglement-assisted advantage in its Petz--\Renyi information; see Figure~\ref{fig:product_entanglement}.

A striking feature of our result is that the strict superadditivity established can even be witnessed without entanglement between the transmitted systems.
In sharp contrast to many celebrated literature whereas quantum nonadditivity is often regarded as an intrinsically entanglement-driven phenomenon \cite{vollbrecht2001entanglement, smith2008quantum, 	hastings2009superadditivity, cubitt2015unbounded, 	elkouss2015superadditivity} \footnote{Note that nonadditivity via classical correlations has also been observed in wiretap classical-quantum channels.}, the witness can be chosen to have a separable, indeed classically correlated, two-copy channel marginal.
This demonstrates that entanglement across the transmitted inputs is not an essential ingredient of the superadditive advantage.








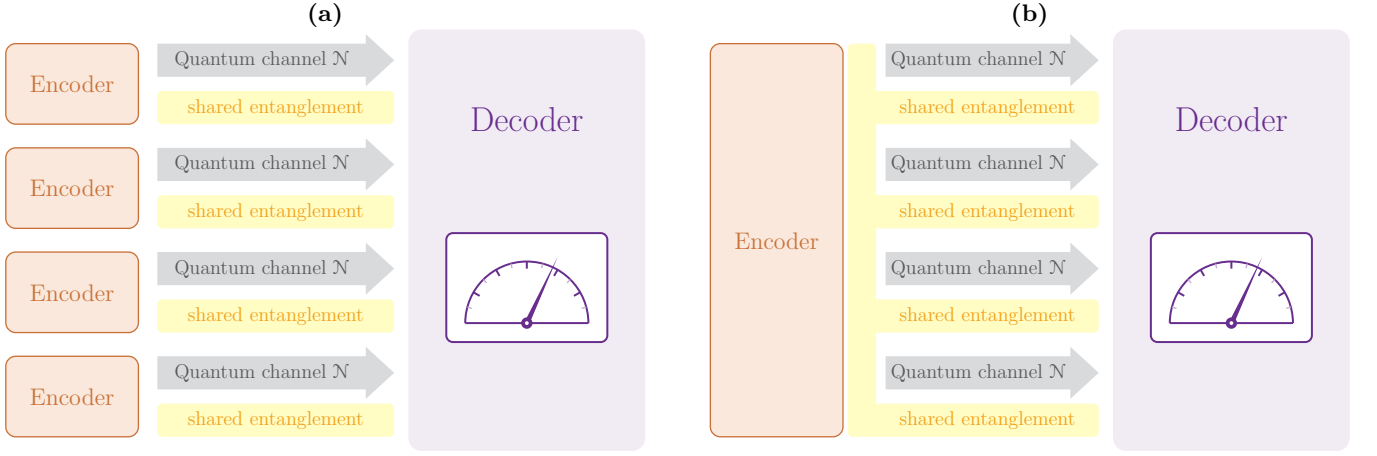
\begin{figure*}[t]
    \centering
    \begin{minipage}{0.48\textwidth}
        \centering
        \textbf{(a)} \\
        \resizebox{\linewidth}{!}{
        \begin{tikzpicture}[x=1cm, y=1cm]

            \definecolor{encFill}{RGB}{249, 232, 219}
            \definecolor{encBorder}{RGB}{202, 114, 57}
            \definecolor{chanFill}{RGB}{217, 218, 220}
            \definecolor{chanText}{RGB}{88, 89, 91}
            \definecolor{entFill}{RGB}{255, 252, 196}
            \definecolor{entText}{RGB}{243, 160, 24}
            \definecolor{decFill}{RGB}{241, 235, 244}
            \definecolor{decText}{RGB}{104, 46, 142}

            \foreach \yOffset in {0, -2.2, -4.4, -6.6} {
                \begin{scope}[yshift=\yOffset cm]
                    \draw[color=encBorder, fill=encFill, thick, rounded corners=6pt] 
                        (0, 0.1) rectangle (2.8, -1.6);
                    \node[color=encBorder, font=\Large] at (1.4, -0.75) {Encoder};

                    \fill[chanFill] 
                        (3.2, 0.1) -- (7.6, 0.1) -- (7.6, 0.3) -- 
                        (8.2, -0.25) -- (7.6, -0.8) -- (7.6, -0.6) -- (3.2, -0.6) -- cycle;
                    \node[color=chanText, font=\large] at (5.4, -0.25) {Quantum channel $\mathscr{N}$};

                    \fill[entFill, rounded corners=3pt] 
                        (3.2, -0.9) rectangle (8.2, -1.6);
                    \node[color=entText, font=\large] at (5.7, -1.25) {shared entanglement};
                \end{scope}
            }

            \fill[decFill, rounded corners=8pt] 
                (8.5, 0.4) rectangle (13.5, -8.5);
            \node[color=decText, font=\huge] at (11.0, -1.5) {Decoder};

            \begin{scope}[shift={(11.0, -5.8)}]
                \draw[color=decText, line width=1.2pt, fill=white, rounded corners=4pt] (-1.7, -0.4) rectangle (1.7, 1.9);
                
                \draw[color=decText, line width=1.2pt] (-1.3, 0) -- (1.3, 0);
                \draw[color=decText, line width=1.2pt] (-1.3, 0) arc (180:0:1.3);
                
                \foreach \angle in {30, 60, 90, 120, 150} {
                    \draw[color=decText, line width=1.2pt] (\angle:1.15) -- (\angle:1.3);
                }
                \foreach \angle in {15, 45, 75, 105, 135, 165} {
                    \draw[color=decText!40, line width=0.8pt] (\angle:1.2) -- (\angle:1.3);
                }
                
                \fill[decText] (-0.06, 0) -- (0.06, 0) -- (65:1.55) -- cycle;
                
                \fill[decText] (0, 0) circle (0.12cm);
                \fill[white] (0, 0) circle (0.04cm);
            \end{scope}

        \end{tikzpicture}
        }
    \end{minipage}
    \hfill 
    \begin{minipage}{0.48\textwidth}
        \centering
        \textbf{(b)} \\
        \resizebox{\linewidth}{!}{
        \begin{tikzpicture}[x=1cm, y=1cm]

            \definecolor{encFill}{RGB}{249, 232, 219}
            \definecolor{encBorder}{RGB}{202, 114, 57}
            \definecolor{chanFill}{RGB}{217, 218, 220}
            \definecolor{chanText}{RGB}{88, 89, 91}
            \definecolor{entFill}{RGB}{255, 252, 196}
            \definecolor{entText}{RGB}{243, 160, 24}
            \definecolor{decFill}{RGB}{241, 235, 244}
            \definecolor{decText}{RGB}{104, 46, 142}

            \draw[color=encBorder, fill=encFill, thick, rounded corners=6pt] 
                (0, 0.1) rectangle (2.8, -8.2);
            \node[color=encBorder, font=\Large] at (1.4, -4.05) {Encoder};

            \foreach \yOffset in {0, -2.2, -4.4, -6.6} {
                \begin{scope}[yshift=\yOffset cm]
                    \fill[entFill, rounded corners=3pt] 
                        (3.3, -0.9) rectangle (8.2, -1.6);
                    \node[color=entText, font=\large] at (5.85, -1.25) {shared entanglement};
                \end{scope}
            }

            \fill[entFill, rounded corners=3pt] 
                (2.9, 0.1) rectangle (3.5, -8.2);

            \foreach \yOffset in {0, -2.2, -4.4, -6.6} {
                \begin{scope}[yshift=\yOffset cm]
                    \fill[chanFill] 
                        (3.7, 0.1) -- (7.6, 0.1) -- (7.6, 0.3) -- 
                        (8.2, -0.25) -- (7.6, -0.8) -- (7.6, -0.6) -- (3.7, -0.6) -- cycle;
                    \node[color=chanText, font=\large] at (5.65, -0.25) {Quantum channel $\mathscr{N}$};
                \end{scope}
            }

            \fill[decFill, rounded corners=8pt] 
                (8.5, 0.4) rectangle (13.5, -8.5);
            \node[color=decText, font=\huge] at (11.0, -1.5) {Decoder};

            \begin{scope}[shift={(11.0, -5.8)}]
                \draw[color=decText, line width=1.2pt, fill=white, rounded corners=4pt] (-1.7, -0.4) rectangle (1.7, 1.9);
                
                \draw[color=decText, line width=1.2pt] (-1.3, 0) -- (1.3, 0);
                \draw[color=decText, line width=1.2pt] (-1.3, 0) arc (180:0:1.3);
                
                \foreach \angle in {30, 60, 90, 120, 150} {
                    \draw[color=decText, line width=1.2pt] (\angle:1.15) -- (\angle:1.3);
                }
                \foreach \angle in {15, 45, 75, 105, 135, 165} {
                    \draw[color=decText!40, line width=0.8pt] (\angle:1.2) -- (\angle:1.3);
                }
                
                \fill[decText] (-0.06, 0) -- (0.06, 0) -- (65:1.55) -- cycle;
                
                \fill[decText] (0, 0) circle (0.12cm);
                \fill[white] (0, 0) circle (0.04cm);
            \end{scope}

        \end{tikzpicture}
        }
    \end{minipage}
    
    \caption{Communication assisted with product and joint entanglement.
    \textbf{(a)} Independent preshared entangled states (highlighted in \colorbox[RGB]{255, 252, 196}{yellow}) are supplied to the individual channel uses. This product resource suffices for the channel capacity $I_1(\mathscr N)$, the sandwiched--\Renyi information $\widetilde I_{\alpha}(\mathscr N)$ for $\alpha\geq\frac12$, and the strong-converse exponent $E_{\mathrm{sc}}(R;\mathscr N)$.
    \textbf{(b)} A joint entangled state is supplied to a joint encoder acting across multiple channel uses. Since the Petz--\Renyi information $I_{\alpha}(\mathscr N)$, $\alpha\in(0,1)$, can be strictly superadditive, joint entanglement can increase the random-coding exponent $E_{\mathrm r}(R;\mathscr N)$ and thereby improve communication quality.}
    \label{fig:product_entanglement}
\end{figure*}

\section{System Model} \label{sec:model}

Let $\mathscr{N}_{\A\to\B}$ be a quantum channel from Alice's system $\A$ to Bob's system $\B$.
In entanglement-assisted (EA) communication, an entangled state $\theta_{\underline{\A}\underline{\R}}$ is shared between Alice (holding $\underline{\A}$) and Bob (holding $\underline{\R}$).
To send a message $m \in \{1,2,\ldots, \lfloor 2^{nR} \rfloor\}$ of rate $R$, Alice applies an encoding operation $\mathscr{E}_{\underline{\A}\to\A^n}^m$ on her part of shared entanglement to prepare a length-$n$ quantum codeword on system $\A^n \equiv \A_1 \ldots \A_n$.
The quantum codeword then undergoes the $n$-fold product channel $\mathscr{N}_{\A\to\B}^{\otimes n}$.
At receiver, Bob applies a quantum measurement $\big\{M_{\underline{\R}\B^n}^{m}\big\}_m$ on the noisy quantum system $\B^n \equiv \B_1 \ldots \B_n$ and his part of shared entanglement.
The decoding error probability is 
\begin{align}
\!\Err\! =\!1-
\frac{1}{\lfloor 2^{nR} \rfloor} \!\sum_{m}	\Tr\!\left[  \mathscr{N}_{\A\to\B}^{\otimes n} \circ \mathscr{E}_{\underline{\A}\to \A^n}^m\!\left(\theta_{\underline{\A}\underline{\R}}\right)
         {M}_{\underline{\R}\B^n}^m \right]\!.
\end{align}
We call such an encoder $\{\mathscr{E}_{\underline{\A}\to \A^n}^m\}_m$, decoder $\big\{M_{\underline{\R}\B^n}^{m}\big\}_m$, and shared entanglement $\theta_{\underline{\A}\underline{\R}}$ an $(n,R,\varepsilon)$ code if $\Err\leq \varepsilon$.
The maximum achievable rate over all EA-codes is
\begin{align}
    R^{\star}(n,\varepsilon)
    \equiv \sup\left\{R: \exists\, (n,R,\varepsilon)\text{ code}\right\},
\end{align}
which is the ultimate \emph{quantity} of information bits Alice can send to Bob with an error tolerance $\varepsilon$.
The well-known BSST theorem \cite{BSS+99, BSS+02, Hol02} characterizes the channel capacity
\begin{align} \label{eq:first-order}
    C_{\EA}(\mathscr{N})
    \equiv \lim_{\varepsilon\to0}\lim_{n\to\infty} R^{\star}(n,\varepsilon)
    = I_1(\mathscr{N}),
\end{align}
in terms of the single-letter mutual information of channel $\mathscr{N}$:
\begin{align} \label{eq:I_1}
  I_1(\mathscr{N})
  =
  \max_{\rho_{\A}} I_1(\R:\B)_{\omega},
  \quad
  \omega_{\R\B}
  =
  (\operatorname{id}_{\R}\otimes\mathscr N)(\psi_{\R\A}^{\rho}),
\end{align}
where $\psi_{\R\A}^{\rho}$ is a purification of the input state $\rho_{\A}$ and $\R\cong\A$.
The quantity $I_1(\mathscr{N})$ is strong additive under tensor product of channels \cite{BSS+02}:
\begin{align} \label{eq:I_1-aditive}
	{I}_{1}(\mathscr{N}_1 \otimes \mathscr{N}_2) = {I}_{1}(\mathscr{N}_1) + {I}_{1}(\mathscr{N}_2).
\end{align}
The channel capacity $C_{\EA}(\mathscr{N})$ is strongly additive as well.

On the other hand, for a given rate $R$, the minimum error probability is defined as
\begin{align}
    \varepsilon^{\star}(n,R)
    \equiv \inf\left\{\varepsilon: \exists\,(n,R,\varepsilon)\text{ code}\right\},
\end{align}
determining the ultimate \emph{quality} of EA-communication.
For transmission rates exceeding capacity, $R > C_{\mathrm{EA}}(\mathscr{N})$, Gupta and Wilde \cite{GW14} and Li and Yao \cite{LY24b} demonstrated that the success probability decays exponentially as
\begin{align}
1- \varepsilon^{\star}(n,R) &\simeq 2^{-n E_{\mathrm{sc}}(R;\mathscr{N})},
\label{eq:strong_converse}
\\
E_{\mathrm{sc}}(R;\mathscr{N})
&\equiv\sup_{\alpha>1} \frac{\alpha-1}{\alpha} \left[ R - \widetilde{I}_{\alpha}(\mathscr{N}) \right].
\label{eq:strong_converse_exponent}
\end{align}
Here, the strong converse exponent is determined by the difference between the rate $R$ and the sandwiched \Renyi information of $\mathscr{N}$ for $\alpha >1$,
\begin{align} \label{eq:defn:sandwiched_info}
\widetilde{I}_{\alpha}(\mathscr{N})
\equiv \max_{\rho_{\A}} \min_{\sigma_{\B}} \widetilde{D}_{\alpha} \left( \mathscr{N}(\psi_{\R\A}^{\rho}) \Vert \rho_{\R} \otimes \sigma_{\B} \right),
\end{align}
where $\widetilde{D}_{\alpha}$ is the sandwiched \Renyi relative entropy \cite{MDS+13, WWY14}.
The sandwiched \Renyi information generalizes $I_1(\mathscr{N})$ to a parametric family and coincides $I_1(\mathscr{N})$ when $\alpha \to 1$.
Moreover, Gupta and Wilde \cite{GW14} proved its additivity:
\begin{align} \label{eq:sadwiched-aditive_>1}
    \widetilde{I}_{\alpha}(\mathscr{N}_1 \otimes \mathscr{N}_2) = \widetilde{I}_{\alpha}(\mathscr{N}_1) + \widetilde{I}_{\alpha}(\mathscr{N}_2), \quad \forall\,\alpha>1,
\end{align}
an underlying property that guarantees a single-letter formula for the strong converse bound and establishes its weak additivity for rates above capacity, i.e.,
\begin{align} \label{eq:E_sc-additive}
    E_{\mathrm{sc}}(2R;\mathscr{N}^{\otimes 2})
    = 2 E_{\mathrm{sc}}(R;\mathscr{N}).
\end{align}

\section{Error exponent and \\Petz--Rényi Information} \label{sec:Petz-Renyi}

For transmission rates below capacity, the convergence rate of the decoding error probability dictates the \emph{quality} of reliable entanglement-assisted communication, making it the central operational quantity of interest. Recently, it was established \cite{preparation} (see prior findings in \cite{QWW18, Cheng_simple}) that
\begin{align}
    \varepsilon^{\star}(n,R)
    &\leq 1.5 \cdot 2^{-E_{\mathrm{r}}(nR;\mathscr{N}^{\otimes n}) },
    \label{eq:random-coding_bound}
    \\
    E_{\mathrm{r}}(R;\mathscr{N}) &\equiv \max_{\sfrac{1}{2}\leq \alpha < 1} \frac{1-\alpha}{\alpha} \left[ I_{\alpha}(\mathscr{N}) - R \right],
    \label{eq:random_coding_exponent}
\end{align}
Fundamentally, it is generally anticipated that the exponent $\frac{1}{n}E_{\mathrm{r}}(nR;\mathscr{N}^{\otimes n})$ to be not merely an achievable bound, but exactly tight, at any transmission rate $R$ above the so-called critical rate $R_{\mathrm{crit}}$, completely characterizing the asymptotic error decay.
This expectation is firmly grounded in its rigorously proven asymptotic tightness for classical channels \cite{SGB67}, classical-quantum channels \cite{Dal13, CHT19}, and covariant quantum channels \cite{SNC26}.

Unlike \eqref{eq:strong_converse_exponent}, the random coding exponent $E_{\mathrm{r}}(R;\mathscr{N})$ is instead characterized by the \emph{Petz--\Renyi information} of order $\alpha \in [\sfrac12,1]$:
\begin{align}
{I}_{\alpha}(\mathscr{N})
&\equiv \max_{\rho_{\A}} \min_{\sigma_{\B}} {D}_{\alpha} \left( \mathscr{N}(\psi_{\R\A}^{\rho}) \Vert \rho_{\R} \otimes \sigma_{\B} \right)
\\
&\!\!\overset{\text{\footnotesize\cite{HT14}}}{=} \max_{\rho_{\A}}  \frac{\alpha}{\alpha-1} \log \Tr\left[ \left( \Tr_{\R}[\rho_{\R}^{1-\alpha} \omega_{\R\B}^{\alpha} ] \right)^{1/\alpha} \right],
\label{eq:objective}
\end{align}
where $D_{\alpha}$ is the Petz--\Renyi relative entropy \cite{Pet86}.
Both $\widetilde{I}_{\alpha}(\mathscr{N})$ and $I_{\alpha}(\mathscr{N})$ are non-decreasing in $\alpha$ and converge to $I_1(\mathscr{N})$ as $\alpha\to1$ \cite{CMWO14}.
The associated exponent functions are depicted in Figure~\ref{fig:error_exponents}.

 \begin{figure}[htbp]
	     \centering
	     \begin{tikzpicture}
		     \begin{axis}[
			         axis lines = left,
			         xlabel = {Rate $R$},
			         ylabel = {Exponents},
			         ymin = 0, ymax = 2.8,
			         xmin = 0, xmax = 4.5,
			         xtick = {1, 2, 3, 4},
			         xticklabels = { , $R_{\mathrm{crit}}$, \quad$C_{\mathrm{EA}}(\mathscr{N}) = I_1(\mathscr{N})$, },
			         ytick = {0.5, 1.0, 1.5, 2.0, 2.5},
			         yticklabels = {\empty}, 
			         grid = major,
			         grid style = {dashed, gray!40},
			         every axis x label/.style={at={(ticklabel* cs:1.02)}, anchor=west},
			         every axis y label/.style={at={(ticklabel* cs:1.05)}, anchor=south},
			         width=8.6cm, height=6.5cm,
			         clip=false,
			         samples=100
			     ]
			
			     \addplot [
			         domain=0:2,
			         thick,
			         blue
			     ] {2.5 - x};
			
			     \addplot [
			         domain=2:3,
			         thick,
			         blue
			     ] {0.5*(3-x)^2};
			
			     \addplot [
			         domain=3:4.5,
			         thick,
			         blue
			     ] {0};
			
			     \addplot [
			         domain=0:3,
			         thick,
			         red
			     ] {0};
			
			     \addplot [
			         domain=3:4.5,
			         thick,
			         red
			     ] {0.5*(x-3)^2};
			
			     \filldraw[blue] (axis cs:2,0.5) circle (1.5pt);
			     \filldraw[black] (axis cs:3,0) circle (1.5pt);
			
			     \node[blue, anchor=south west] at (axis cs:0.2, 2.3) {$E_{\mathrm{r}}(R;\mathscr{N})$};
			     \node[red, anchor=north west] at (axis cs:3.4, 1.3) {$E_{\mathrm{sc}}(R;\mathscr{N})$};
			
			     \end{axis}
		     \end{tikzpicture}
	     \caption{Illustration of the random coding exponent $E_{\mathrm{r}}(R;\mathscr{N})$ in \eqref{eq:random_coding_exponent} and the strong converse exponent $E_{\mathrm{sc}}(R;\mathscr{N})$ in \eqref{eq:strong_converse_exponent} a function of the rate $R$; both indicate the quality of entanglement-assisted communication.
		     The latter $E_{\mathrm{sc}}(R;\mathscr{N})$ is additive \cite{GW14}, while we show that $E_{\mathrm{r}}(R;\mathscr{N})$ can be strictly superadditive for any $R<I_1(\mathscr{N})$.
		     Here, $I_1(\mathscr{N})$ is additive and quantifies the maximum achievable rate with vanishing errors.}
	     \label{fig:error_exponents}
	 \end{figure}
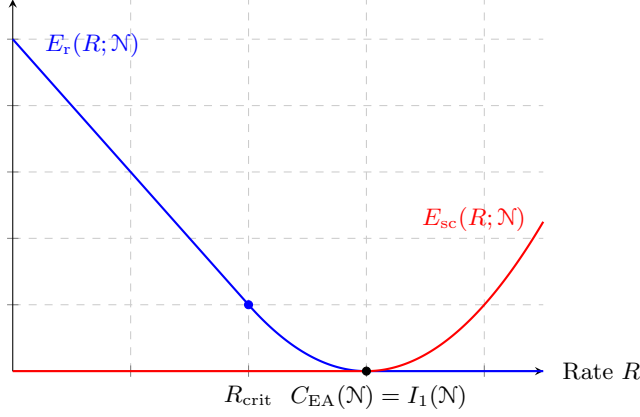

Given that joint entanglement across multiple channel uses do not increase the channel capacity, i.e., the  additivity of $I_1(\mathscr{N})$ in \eqref{eq:I_1-aditive}, and recalling the additivity of $\widetilde{I}_{\alpha}(\mathscr{N})$ in \eqref{eq:sadwiched-aditive_>1} for $\alpha>1$ and for $\alpha\in[\sfrac12,1)$ by Li--Xu \cite{li2026completely},
it is natural to expect the following question:
\begin{align} \tag{Q1} \label{eq:Q:supperadditivity} 
	\textit{Is Petz--\Renyi $I_{\alpha}(\mathscr{N})$, $\alpha \in [\sfrac12,1)$, additive too?}
\end{align}
Operationally, this corresponds to the question:
\begin{align} \tag{Q2} \label{eq:Q:quality} 
	\begin{minipage}{0.42\textwidth}
	\textit{{Can joint correlations across multiple channel uses enhance the quality of EA-communication?}}
\end{minipage}
\end{align}

Indeed, we show that $I_{\alpha}(\mathscr{N})$ is equivalent to a convex optimization, as opposed to the non-convex optimizations of the
(unassisted) Holevo capacity $\chi(\mathscr{N})$ and  quantum channel capacity $Q(\mathscr{N})$.

\begin{proposition}[Convex optimization reduction]  \label{prop:concavity}
$I_{\alpha}(\mathscr{N})$ given in \eqref{eq:objective} is equivalent to a convex optimization.
\end{proposition}
\noindent We defer the detailed proof to Appendix~\ref{sec:properties}. Moreover, additivity of $I_{\alpha}(\mathscr{N})$ does indeed hold for some channels.

\begin{proposition}[Additivity for special channels] \label{prop:additivity_special}
$I_{\alpha}(\mathscr{N})$, $\alpha \in (0,1)$, is strongly additive for\vspace{-6pt}
\begin{enumerate}
    \item isometric channels,\vspace{-6pt}
    \item projective measurement channels,\vspace{-6pt}
    \item covariant quantum channels,\vspace{-6pt}
    \item and classical-quantum channels.
\end{enumerate}
\end{proposition}
\noindent We defer the detailed proof to Appendix~\ref{sec:properties}.

\section{Strict Superadditivity}

In quantum information, a favorable convex optimization landscape typically ensures that fundamental quantities remain computable via single-letter formulas.
Unexpectedly, the Petz--\Renyi information $I_{\alpha}(\mathscr{N})$ possesses fundamentally distinct behaviors from the sandwiched \Renyi information $\widetilde{I}_{\alpha}(\mathscr{N})$.
Indeed, we prove that $I_{\alpha}(\mathscr{N})$ can exhibit strict superadditivity for any $\alpha \in [\sfrac12,1)$, which falsifies \eqref{eq:Q:supperadditivity} in general.
This finding answers \eqref{eq:Q:quality} in the affirmative---although joint preshared entanglement does not increase the ultimate \emph{quantity} of EA-communication, it in general enhance the communication \emph{quality} for certain channels.
Namely, as opposed to the additivity of the strong converse exponent in \eqref{eq:E_sc-additive}, 
\begin{align} \label{eq:E_r-supperadditive}
	\exists\, \mathscr{N}:
	E_{\mathrm{r}}(2R;\mathscr{N}^{\otimes 2})
	> 2 E_{\mathrm{r}}(R;\mathscr{N}), \quad \forall\, R<C_{\mathrm{EA}}(\mathscr{N}).
\end{align}

\emph{Entanglement-breaking}. Our first example is the single-heavy Fourier measurement $\mathscr{M}_{\A\to \Y}$ with rank-one effects
	\begin{align}  \label{eq:Fourier_POVM0}
		M_{\A}^y&=\frac{1}{\lambda d_{\A}} |v_y\rangle\!\langle v_y|,
		\qquad
		y=0,\ldots,d_{\A}-1,
	\end{align}
	where $\lambda\in(1/d_{\A},1)$, $a_0=\lambda$, $a_1=\cdots=a_{d_{\A}-1}=\frac{1-\lambda}{d_{\A}-1}$,
	$|v_y\rangle
	=
	\sum_{j=0}^{d_{\A}-1}\sqrt{a_j}\,\mathrm{e}^{2\pi \mathrm{i}jy/d_{\A}}|j\rangle$, and the residual effect $M_{\A}^{\infty} = \I_{\A}-\sum_{j=0}^{d_{\A}-1} M_{\A}^j\geq 0_{\A}$.

\begin{theorem}[Strict superadditivity for measurement channels]
	\label{thm:strict-superadditivity}
	For every \(d_{\A}\ge3\), every
	\(\lambda\in(1/d_{\A},1)\), and every \(0<\alpha<1\), the Fourier measurement defined in \eqref{eq:Fourier_POVM0}  satisfies
	\begin{align}
		I_{\alpha}(\mathscr{M}^{\otimes 2}) > 2 I_{\alpha}(\mathscr{M})
		\qquad
		\forall\,0<\alpha<1.
	\end{align}
\end{theorem}
\noindent We defer the detailed proof to Appendix~\ref{sec:q-c}.

A numerical certificated superadditivity is in Figure~\ref{fig:Fourier_POVM}.
\begin{figure}[ht]
	\centering
	\includegraphics[width=1\columnwidth]{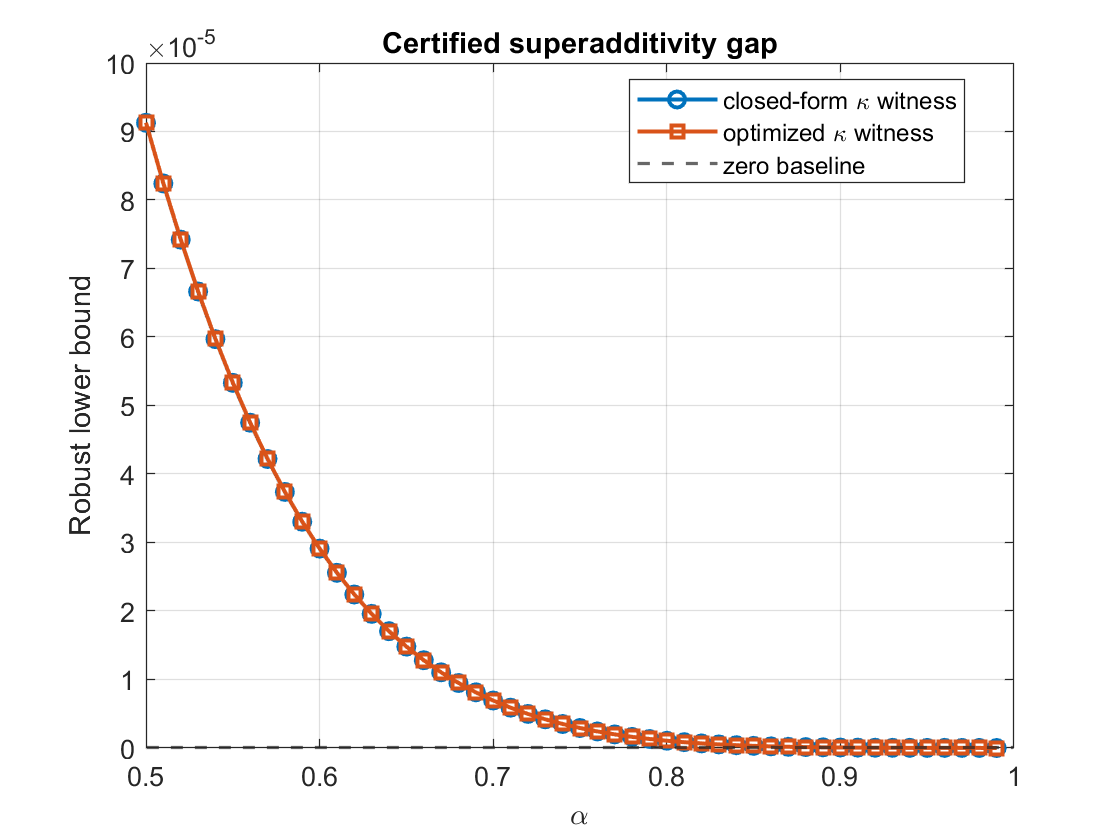}
	\caption{Certifying numerical superadditivity via the single-heavy Fourier measurement defined in \eqref{eq:Fourier_POVM0}  on $\mathbb{C}^4$.
    The $y$-axis plots the lower bound on $I_{\alpha}(\mathscr{M}^{\otimes 2}) - 2 I_{\alpha}(\mathscr{M})$.
	The choices of the $\kappa$ in \eqref{eq:pi-kappa0} for the two-copy case is according to the quadratic approximation in \eqref{eq:pi-kappa0-minimizing}.
	}
	\label{fig:Fourier_POVM}
\end{figure}

Our second example is the amplitude damping channels with Choi matrix:
\begin{align} \label{eq:GAD}
	\Gamma_{\A\B}^{\mathscr{N}}=\begin{pmatrix}
		1&0&0&\sqrt{1-\gamma}\\
		0&0&0&0\\
		0&0&\gamma&0\\
		\sqrt{1-\gamma}&0&0&1-\gamma
	\end{pmatrix}.
\end{align}

\begin{theorem}[Strict superadditivity for amplitude damping channels]
	For the amplitude damping channel $\mathscr{N}$ given in \eqref{eq:GAD} with \(0<\gamma<1\), one has
	\begin{align}
		I_\alpha(\mathscr{N}^{\otimes 2}) > 2 I_{\alpha}(\mathscr{N})
		\qquad
		\forall\,0<\alpha<1.
	\end{align}
\end{theorem}
\noindent We defer the detailed proof to Appendix~\ref{sec:GAD}.

To analytically prove the strict superadditivity in the above two examples, we first employ Proposition~\ref{prop:concavity} to show that the one-copy Petz--\Renyi reduces to a single-parameter convex optimization and the one-copy optimizer $\rho_{\A}^{\star}$ can be chosen as diagonal.
Denoting the trace functional in \eqref{eq:objective} by 
$Q_{\alpha}(\rho) \equiv \Tr\big[ \big( \Tr_{\R}[\rho_{\R}^{1-\alpha} \omega_{\R\B}^{\alpha} ] \big)^{1/\alpha} \big]$,
we choose a correlated diagonal two-copy ansatz as
\begin{align} \label{eq:pi-kappa0}
	\rho_{\A_1\A_2}(\kappa) = \rho_{\A_1}^{\star} \otimes \rho_{\A_2}^{\star} + \kappa \Delta,
\end{align}
where $\Delta$ is a diagonal traceless operator and $\kappa$ is sufficiently small.
For any $\alpha \in (0,1)$, we prove that
\begin{align}
	c_{1}(\alpha) \equiv \left.\frac{\mathrm{d}}{\mathrm{d}\kappa} Q_{\alpha}\left(	\rho_{\A_1\A_2}(\kappa) \right) \right|_{\kappa=0} < 0.
\end{align}
By Taylor's expansion:
\begin{align} \label{eq:pi-kappa0-minimizing}
	Q_{\alpha}\left(\rho_{\A_1\A_2}(\kappa)\right)
	= Q_{\alpha}\left(\rho_{\A}^\star\right)^2 + c_{1}(\alpha) \kappa + \mathcal{O}(\kappa^2),
\end{align}
the strict superadditivity is then witnessed by the infinitesimal diagonal classically
correlated path \(\rho_\alpha(\kappa)\).

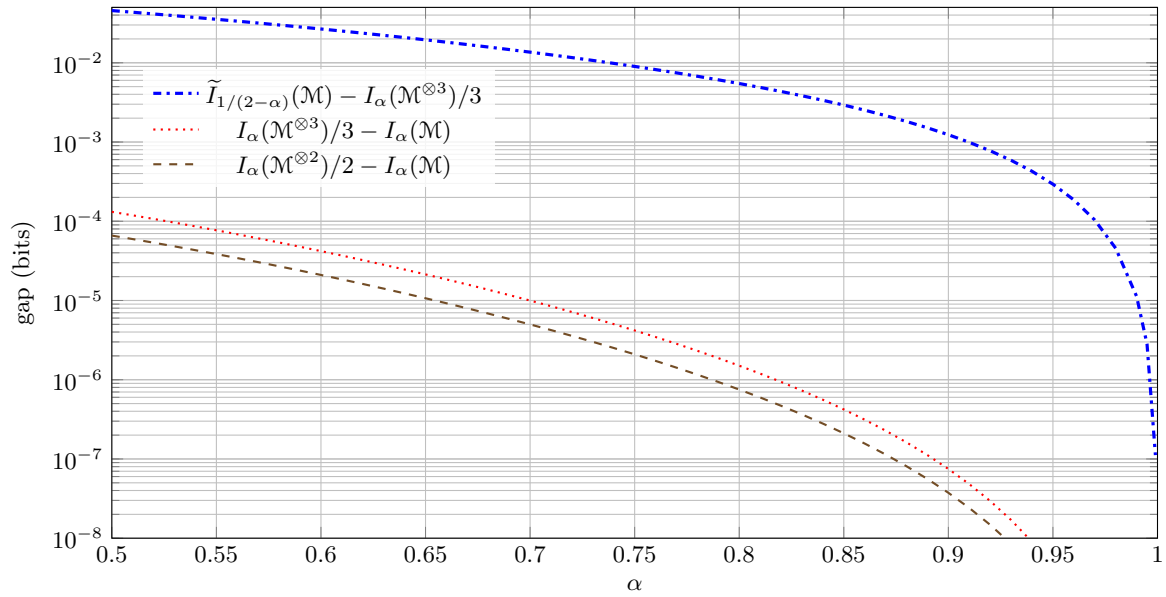
\begin{figure*}[th!]
\centering

\def\NatsToBits{1.4426950408889634}

\begin{tikzpicture}
\begin{semilogyaxis}[
  width=0.86\textwidth,
  height=0.48\textwidth,
  xmin=0.5,xmax=1.0,
  ymin=1e-8,ymax=5e-2,
  grid=both,
  xlabel={\(\alpha\)},
  ylabel={gap (bits)},
  legend style={
    at={(0.03,0.97)},
    anchor=north west,
    yshift=-6mm,
    draw=none,
    fill=white,
    fill opacity=0.85,
    text opacity=1
  },
]

\addplot+[
  mark=none,
  very thick,
  dashdotted
] table[
  x=alpha,
  y expr={\NatsToBits*\thisrow{gapSand}}
] {gap_data.dat};

\addlegendentry{
  \(\widetilde I_{1/(2-\alpha)}(\mathscr M)
  -I_\alpha(\mathscr M^{\otimes3})/3\)
}

\addplot+[
  mark=none,
  thick,
  dotted
] table[
  x=alpha,
  y expr={\NatsToBits*\thisrow{gapThree}}
] {gap_data.dat};

\addlegendentry{
  \(I_\alpha(\mathscr M^{\otimes3})/3-I_\alpha(\mathscr M)\)
}

\addplot+[
  mark=none,
  thick,
  dashed
] table[
  x=alpha,
  y expr={\NatsToBits*\thisrow{gapTwo}}
] {gap_data.dat};

\addlegendentry{
  \(I_\alpha(\mathscr M^{\otimes2})/2-I_\alpha(\mathscr M)\)
}

\end{semilogyaxis}
\end{tikzpicture}

\caption{The single-letter sandwiched upper gap decreases to zero as
\(\alpha\nearrow1\).}
\label{fig:gap-plot}
\end{figure*}

\section{Single-Letter Bounds}

Since $I_{\alpha}\left(\mathscr{N}\right)$ can be strictly superadditive, the largest Petz channel information in the asymptotic limit is expressed by its regularization:
\begin{align}
	I_{\alpha}^{\infty}(\mathscr{N})
	&\coloneq \sup_{n\in\mathds{N}} \frac{1}{n} I_{\alpha}\left(\mathscr{N}^{\otimes n}\right)
	= \lim_{n\to\infty} \frac{1}{n} I_{\alpha}\left(\mathscr{N}^{\otimes n}\right).
\end{align}
While this regularization captures the optimal communication quality, it unfortunately demands an intractable, infinite-dimensional optimization. To determine the ultimate limits of this multi-copy advantage, we bypass this incomputability by deriving a computable single-letter upper bound. See Figure~\ref{fig:gap-plot} for the numerical results.

\begin{proposition}[Single-letter bounds]
	For any quantum channel $\mathscr{N}_{\A\to\B}$ and $\alpha \in (0,1)$,
	\begin{align}
		I_{\alpha}(\mathscr{N})
		\leq 
		I_{\alpha}^{\infty}(\mathscr{N})
		\leq
		\widetilde{I}_{\frac{1}{2-\alpha}}(\mathscr{N})
		\leq 
		{I}_{\frac{1}{2-\alpha}}(\mathscr{N}),
	\end{align}
	where the upper bounds both converge to $I_{1}(\mathscr{N})$ as $\alpha \to 1$.
\end{proposition}
\noindent We defer the detailed proof to Appendix~\ref{sec:single-letter}.

The single-letter upper bound in terms of $\widetilde{I}_{\frac{1}{2-\alpha}}(\mathscr{N})$ is a double-state optimization.  One may further relax it to the Petz version.
We show in Appendix~\ref{sec:single-letter} that the sandwiched \Renyi information still admits a one-state optimization expression for quantum-classical channels.

\section{Discussion}

In this paper, we analytically prove the strict superadditivity of the Petz--\Renyi information $I_{\alpha}(\mathscr{N})$ of order $\alpha \in [\sfrac12,1)$ for a family of quantum-classical channels and the amplitude damping channels.
This implies that joint preshared entanglement can increase the random coding exponent for transmission rates below capacity.
Our proof extends to $\alpha \in(0,1)\cup(1,2)$ as well, and the strict superadditivity vanishes at $\alpha = 1$ \cite{BSS+02} and $\alpha = 2$. See Table~\ref{tab:renyi_additivity} for the summary.

One may wonder if a stronger resource for assisting communication would ease the strict superadditivity of $I_{\alpha}(\mathscr{N})$.
In fact, Ref.~\cite{ATB24} considers non-signaling assistance with one-bit forward activation, and the resulting error exponent is given by the same $I_{\alpha}(\mathscr{N})$ of order $\alpha \in (0,1)$.
Our results then imply that joint assisting resource across multiple channel uses can still enhance the error exponent.
While Girardi \textit{et al.}~demonstrated that the zero-rate error exponent admits a regularized expression \cite{Qumlaut}, our investigation targets the constant-rate regime near capacity—the operational domain most crucial for evaluating practical, high-throughput  communication.

Finally, numerical evidence from our single-heavy Fourier measurement and amplitude-damping channel also demonstrate strict superadditivity for the Petz--\Renyi channel entropy, an additivity question posed by Gour and Wilde \cite{GW21}. Our examples then show that product strategies for channel discrimination in \cite{CMWO14} are suboptimal in general.


 \begin{table}[htbp]
	     \centering
	     \caption{Additivity properties of the Petz--\Renyi information $I_\alpha(\mathscr{N})$ and the sandwiched \Renyi information $\widetilde{I}_\alpha(\mathscr{N})$ across different ranges of $\alpha$.
		     We show that $I_\alpha(\mathscr{N})$ can be strictly superadditive for $\alpha \in (0,1)\cup(1,2)$. 
             Ranges of $\alpha$ marked with '$\times$' are excluded, as they fail to satisfy the data-processing inequality.
		     }
	     \label{tab:renyi_additivity}
	
	     \renewcommand{\arraystretch}{1.8}

	     \arrayrulecolor{black}
	     \begin{adjustbox}{max width=\columnwidth}
		     \begin{tabular}{l !{\color{lightgray}\vrule} c !{\color{lightgray}\vrule} c !{\color{lightgray}\vrule} c !{\color{lightgray}\vrule} c !{\color{lightgray}\vrule} c !{\color{lightgray}\vrule} c}
			         \toprule
			         $\alpha$ & $(0, \frac12)$ & $[\frac12, 1)$ & $1$ & $(1, 2)$ & $2$ & $(2, +\infty]$ \\
			         \midrule
			
			         $I_\alpha$ &
			         \multicolumn{2}{c!{\color{lightgray}\vrule}}{{\renewcommand{\arraystretch}{1.0}\begin{tabular}{@{}c@{}}Non-additive \\ (Theorem~\ref{thm:strict-superadditivity})\end{tabular}}} &
			         \multirow{2.5}{*}{{\renewcommand{\arraystretch}{1.0}\begin{tabular}{@{}c@{}}Additive \\ \cite{BSS+02}\end{tabular}}} &
			         {\renewcommand{\arraystretch}{1.0}\begin{tabular}{@{}c@{}}Non-additive \\ (App.~\ref{sec:q-c})\end{tabular}} &
			         {\renewcommand{\arraystretch}{1.0}\begin{tabular}{@{}c@{}}Additive \\ (App.~\ref{sec:properties})\end{tabular}} &
			         $\times$ \arrayrulecolor{gray}\\
			
			         \cmidrule{2-3} \cmidrule{5-7}
			
			         \arrayrulecolor{black}$\widetilde{I}_\alpha$ &
			         $\times$ &
			         {\renewcommand{\arraystretch}{1.0}\begin{tabular}{@{}c@{}}Additive \\ \cite{li2026completely}\end{tabular}} &
			         &
			         \multicolumn{3}{c}{{\renewcommand{\arraystretch}{1.0}\begin{tabular}{@{}c@{}}Additive \\ \cite{GW14}\end{tabular}}} \\
			
			         \bottomrule
			     \end{tabular}
		     \end{adjustbox}
	 \end{table}


\bibliography{reference, ref, ref-additivity}
\onecolumngrid

\clearpage 
\appendix
\setcounter{secnumdepth}{2}

\section{Definitions and Notation} \label{sec:definition}

We consider finite-dimensional Hilbert space.
We denote by $\rho_{\A}$ and $\sigma_{\B}$ quantum states (i.e.~density matrix) on quantum systems $\A$ and $\B$, respectively.
The optimizations $\max_{\rho_{\A}}$ or $\min_{\sigma_{\B}}$ are over the state space on systems $\A$ or $\B$.
We drop the subscript $\A$ and $\B$ if the name of the quantum systems are irrelevant.
We use $\R$ to stand for the reference system of $\A$; hence, $\R\cong \A$.
For $p\geq 1$, we define the Schatten norm $\|X\|_p \equiv (\Tr[|X|^p])^{1/p}$.

A quantum channel $\mathscr{N}_{\A\to\B}$ is a completely positive and trace-preserving map from system $\A$ to system $\B$.
We denote its unnormalized Choi operator by
\begin{align}
    \Gamma_{\A\B}^{\mathscr{N}}
    &\equiv \mathrm{id} \otimes \mathscr{N}_{\bar{\A}\to\B}(|\tilde{\Phi}\rangle\langle\tilde{\Phi}|_{\A\bar{\A}}),
    \\
    |\tilde{\Phi}\rangle_{\A\bar{\A}} &\equiv \sum_i |i\rangle_{\A} |i\rangle_{\bar{\A}}, \quad \bar{\A} \cong \A.
    \label{eq:unnormalized_MES}
\end{align} 
The joint input-output state given an input $\rho_{\A}$ is denoted by
\begin{align}
	\omega_{\R\B}^{\rho} 
	\equiv \sqrt{\rho_{\R}} \Gamma_{\R\B}^{\mathscr{N}} \sqrt{\rho_{\R}}.
\end{align}
The superscript $\rho$ will be dropped if the context is clear.
We use 
\begin{align}
\mathscr{M}_{\A\to\Y}(\rho_{\A})
= \sum_{y\in\Y} \Tr\left[\rho_{\A} M_{\A}^y\right] |y\rangle\langle y|_{\Y}
\end{align}
for a quantum-classical (measurement) channel, described by the associated positive operator-valued measure (POVM) $\{M_{\A}^y\}_{y\in\Y}$.

For quantum states $\rho$ and $\sigma$ and $\alpha \in (0,1)\cup(1,+\infty)$, we define the Petz \cite{Pet86} and sandwiched \cite{MDS+13, WWY14} \Renyi relative entropies, respectively, as
\begin{align}
	D_{\alpha}(\rho\Vert\sigma)
	&\equiv \frac{1}{\alpha-1} \log \Tr\left[ \rho^{\alpha} \sigma^{1-\alpha} \right],
	\\
	\widetilde{D}_{\alpha}(\rho\Vert\sigma)
	&\equiv
	\frac{1}{\alpha-1} \log \Tr\left[ \left( \sigma^{\frac{1-\alpha}{2\alpha}} \rho \sigma^{\frac{1-\alpha}{2\alpha}} \right)^{\alpha} \right].
\end{align}
For $\alpha < 1$, both quantities are defined to be infinite for orthogonal states.
For $\alpha > 1$, both quantities are defined for $\text{supp} (\rho) \subseteq \text{supp}(\sigma)$, and infinite otherwise.
The end points $\alpha \in \{0,1,+\infty\}$ are defined by continuous extension.
It is known that $\widetilde{D}_{\alpha}(\rho\Vert\sigma) \leq D_{\alpha}(\rho\Vert\sigma)$ \cite{Ara90, MDS+13, MO14};
both quantities are non-decreasing in $\alpha$ and 
\begin{align}
    \lim_{\alpha\to 1} \widetilde{D}_{\alpha}(\rho\Vert\sigma) 
    = \lim_{\alpha\to 1} {D}_{\alpha}(\rho\Vert\sigma) 
    = D_1(\rho\Vert\sigma) \equiv \Tr\left[\rho\left(\log \rho - \log \sigma\right)\right].
\end{align}
The Petz--\Renyi relative entropy is contractive under any quantum channel for $\alpha \in [0,2]$ \cite{Pet86}, while sandwiched \Renyi relative entropy satisfies this property for $\alpha \in [\sfrac12,+\infty]$ \cite{MDS+13, WWY14, MO14}.

Define the Petz--\Renyi information and sandwiched \Renyi information of $\mathscr{N}$ as
\begin{align} 
	{I}_{\alpha}(\mathscr{N})
	&\equiv \sup_{\rho_{\A}} \min_{\sigma_{\B}} {D}_{\alpha} \left( \mathscr{N}(\psi_{\R\A}^{\rho}) \Vert \rho_{\R} \otimes \sigma_{\B} \right)
	\label{eq:Pez-Renyi_defn}
	\\
	&= \sup_{\rho_{\A}}  \frac{\alpha}{\alpha-1} \log \Tr\left[ \left( \Tr_{\A}[\rho_{\A}^{1-\alpha} (\sqrt{\rho_{\A}} \Gamma_{\A\B}^{\mathscr{N}} \sqrt{\rho_{\A}})^{\alpha} ] \right)^{1/\alpha} \right],\label{eq:Petz-Renyi_app}
	\\
	\widetilde{I}_{\alpha}(\mathscr{N})
	&\equiv \sup_{\rho_{\A}} \min_{\sigma_{\B}} \widetilde{D}_{\alpha} \left( \mathscr{N}(\psi_{\R\A}^{\rho}) \Vert \rho_{\R} \otimes \sigma_{\B} \right),
	\label{eq:defn_sandwiched-app}
\end{align}
where $\psi_{\R\A}^{\rho}$ is a purification of the input state $\rho_{\A}$.
Equality \eqref{eq:Petz-Renyi_app} follows from \cite{HT14}.
Note that for $\alpha \in (0,1)$, both quantities are finite; we may change $\sup_{\rho_{\A}}$ to $\max_{\rho_{\A}}$.
By the relations between the Petz and sandwiched \Renyi relative entropies, we have \cite{CMWO14}
\begin{align}
	\widetilde{I}_{\alpha}(\mathscr{N})
	\leq I_{\alpha}(\mathscr{N}), \quad \forall\, \alpha \geq 0;
    \quad
    \lim_{\alpha\to1} \widetilde{I}_{\alpha}(\mathscr{N})
    =  \lim_{\alpha\to1} {I}_{\alpha}(\mathscr{N})
    = {I}_{1}(\mathscr{N}).
\end{align}

The \emph{additivity} notions of an extended-real-valued function $f$ on the set of quantum channels are the following:
\begin{align}
    &\text{(strongly additive)} &&f(\mathscr{N}_1\otimes\mathscr{N}_2)
    = f(\mathscr{N}_1) + f(\mathscr{N}_2);
    \\
    &\text{(weakly additive)} &&f(\mathscr{N}\otimes\mathscr{N})
    = 2 f(\mathscr{N});
    \\
    &\text{(superadditive)}
    &&f(\mathscr{N}_1\otimes\mathscr{N}_2)
    \geq f(\mathscr{N}_1) + f(\mathscr{N}_2);
    \\
    &\text{(subadditive)}
    &&f(\mathscr{N}_1\otimes\mathscr{N}_2)
    \leq f(\mathscr{N}_1) + f(\mathscr{N}_2).
\end{align}

\section{Properties of The Petz--Rényi Information} \label{sec:properties}
In this section, we derive basic properties of the Petz--\Renyi information $I_{\alpha}(\mathscr{N})$.
We first show that $I_{\alpha}(\mathscr{N})$ is equivalent to a convex optimization (Proposition~\ref{prop:Petz_convexity}).
Second, $I_{\alpha}(\mathscr{N})$ is strongly additive for some special channels (Proposition~\ref{prop:additivity_special-app}).
Third, $I_{\alpha}(\mathscr{N})$ is strongly additive for any channel at $\alpha = 2$ (Proposition~\ref{prop:alpha-2}).

Though the objective function on the right-most side of \eqref{eq:Petz-Renyi_app} is not concave in $\rho_{\A}$ for $\alpha \in (0,1)$,
we can consider the minimization of the map
$\rho_{\A} \mapsto \Tr\left[ \left( \Tr_{\A} \left[ \rho_{\A}^{1-\alpha} \left( \sqrt{\rho_{\A}} \Gamma_{\A\B}^{\mathscr{N}} \sqrt{\rho_{\A}} \right)^{\alpha}  \right] \right)^{1/\alpha} \right]$ for $\alpha \in (0,1)$, since the logarithm is monotone.
The following Proposition~\ref{prop:Petz_convexity} shows the convexity, which in turn, implies that the objective function \eqref{eq:Petz-Renyi_app} is quasi-concave in $\rho_{\A}$ for any $\alpha \in (0,1)$.

\begin{app-proposition}[Convex optimization reduction] \label{prop:Petz_convexity}
	Let $\rho_{\A}$ be a state and let $\Gamma_{\A\B}^{\mathscr{N}}$ be the Choi operator of a channel $\mathscr{N}_{\A\to \B}$.
	Then,
	the map
	\begin{align} \label{eq:app-objective}
		\rho_{\A} \mapsto \Tr\left[ \left( \Tr_{\A} \left[ \rho_{\A}^{1-\alpha} \left( \sqrt{\rho_{\A}} \Gamma_{\A\B}^{\mathscr{N}} \sqrt{\rho_{\A}} \right)^{\alpha}  \right] \right)^{1/\alpha} \right]
	\end{align}
	on density operators is convex for any $\alpha \in (0,1)$ and concave for any $\alpha \in (1,2\,]$.
\end{app-proposition}

\begin{proof}
    Write $\Gamma_{\A\B}=\Gamma_{\A\B}^{\mathscr{N}}$.
    Without loss of generality, we only prove the case of $\rho_{\A} > 0$ and $\Gamma_{\A\B}>0$.
	The case of $\rho_{\A} \geq 0$ and $\Gamma_{\A\B} \geq 0$ follows from substituting $\rho_{\A} \leftarrow (1-\epsilon) \rho_{\A} + \epsilon \I_{\A}/d_{\A}$, $\Gamma_{\A\B}\leftarrow (1-\epsilon)\Gamma_{\A\B} + \epsilon \I_{\A}\otimes\I_{\B}/d_{\B}$, continuity, and letting $\epsilon \searrow 0$.
    
    First, consider $\alpha \in (0,1)$.
	It is sufficient to show the convexity of the map
	\begin{align}
		\rho_{\A} \mapsto \Tr_{\A} \left[ \rho_{\A}^{1-\alpha} \left( \sqrt{\rho_{\A}} \Gamma_{\A\B} \sqrt{\rho_{\A}} \right)^{\alpha}  \right] 
	\end{align}
	on positive semi-definite operators, since $\Tr[ \left(\cdot\right)^{1/\alpha} ]$ is convex and non-decreasing for $\alpha \in (0,1)$.
		
	Via polar decomposition, we have $ L f(L^\dagger L) = f( L L^\dagger) L $ for any unitary-invariant functional calculus $f$.
	Then, 
	\begin{align}
		\left( \sqrt{\rho_{\A}} \Gamma_{\A\B} \sqrt{\rho_{\A}} \right)^{\alpha}
		= \sqrt{\rho_{\A}} \sqrt{\Gamma_{\A\B}} \big( \sqrt{\Gamma_{\A\B}} \rho_{\A} \sqrt{\Gamma_{\A\B}} \big)^{\alpha-1} \sqrt{\Gamma_{\A\B}} \sqrt{\rho_{\A}},
	\end{align}
	which means that it is equivalent to consider the map
	\begin{align}
		\rho_{\A} 
		\mapsto &\Tr_{\A}\left[ \sqrt{\Gamma_{\A\B}} \big( \sqrt{\Gamma_{\A\B}} \rho_{\A} \sqrt{\Gamma_{\A\B}} \big)^{\alpha-1} \sqrt{\Gamma_{\A\B}} \rho_{\A}^{2-\alpha}  \right]
		\\
		&= \langle \tilde{\Phi} |_{\A \bar{\A}} \sqrt{\Gamma_{\A\B}} \big( \sqrt{\Gamma_{\A\B}} \rho_{\A} \sqrt{\Gamma_{\A\B}} \big)^{\alpha-1} \sqrt{\Gamma_{\A\B}} \otimes (\rho_{\bar{\A}}^{\top})^{2-\alpha} 
		|\tilde{\Phi} \rangle_{\A \bar{\A}}, \label{eq:concavity_end}
	\end{align}
	where $|\tilde{\Phi}\rangle_{\A \bar{\A}}$ was defined in \eqref{eq:unnormalized_MES}.
	
	Now, invoke Lemma~\ref{lemm:Lieb-Ando} below with $X \leftarrow \sqrt{\Gamma_{\A\B}} \rho_{\A} \sqrt{\Gamma_{\A\B}}$, $Y \leftarrow \rho_{\bar{\A}}^{\top}$, and $p \leftarrow \alpha - 1 \in (-1,0)$.
	Then, the map given in \eqref{eq:concavity_end} is convex by noting that the map $\rho_{\A} \mapsto \sqrt{\Gamma_{\A\B}} \rho_{\A} \sqrt{\Gamma_{\A\B}}$ is linear, transpose $(\cdot)^{\top}$ is linear, and $\langle \Phi |_{\A \bar{\A}}  \sqrt{\Gamma_{\A\B}} (\cdot) \sqrt{\Gamma_{\A\B}} \otimes \I_{\bar{\A}} |\Phi \rangle_{\A \bar{\A}}$ is a positive map.
 
    The proof of the case $\alpha \in (1,2)$ follows similarly by noting that $\Tr[ \left(\cdot\right)^{1/\alpha} ]$ is concave and non-decreasing for $\alpha \in (1,2)$.
    For $\alpha = 2$, we take the pointwise limit $\alpha\nearrow 2$ of the concavity for $\alpha \in (1,2)$.
    
    \begin{app-lemma}[Lieb's Concavity Theorem {\cite{Lie73}} \& Ando's Convexity Theorem {\cite{And79}}] \label{lemm:Lieb-Ando}
	Let $X,Y > 0$ be positive operators.
	The map $(X,Y)\mapsto X^{p} \otimes Y^{1-p}$ on positive definite operators is jointly convex for $p \in (-1,0) \cup (1,2)$ and 
	is jointly concave for $p \in (0,1)$.
	%
\end{app-lemma}
\end{proof}

\begin{app-proposition}[Additivity for special channels] \label{prop:additivity_special-app}
The following expressions and strong additivity hold for $I_{\alpha}(\mathscr{N})$.
\begin{enumerate}
    \item\label{app-item:isometric} Isometric channels: 
    \begin{align}
    I_{\alpha}(\mathscr{N}) = \sup_{\rho_{\A}} 2 H_{\frac{2-\alpha}{\alpha}}(\A)_{\rho} =
    \begin{dcases}
        2 \log d_{\A} & \alpha \in (0,2],
        \\
        +\infty & \alpha > 2,
    \end{dcases}
    \end{align}
    where $H_{\alpha}(\A)_{\rho} \equiv \frac{1}{1-\alpha}\log \Tr[\rho^{\alpha}]$ is the \Renyi entropy and
    $d_{\A}$ denotes the dimension of the input Hilbert space.
    
    \item\label{app-item:projective} Projective measurement channels with $|\Y|$ nonzero projectors: $I_{\alpha}(\mathscr{N}) = \log |\Y|$ for $\alpha \in (0,2]$.
    
    \item\label{app-item:covariant} Covariant quantum channels \footnote{We impose the standard condition on the group covariant channels whose representation acting
on the input space is irreducible.}:
    $I_{\alpha}(\mathscr{N}) = \frac{1}{1-\alpha}\log d_{\A} + \frac{\alpha}{\alpha-1} \log \Tr\left[ \left( \Tr_{\A}\left[ (\Gamma_{\A\B}^{\mathscr{N}})^{\alpha} \right] \right)^{1/\alpha} \right]$ for $\alpha \in (0,1)\cup(1,2]$.

    \item\label{app-item:commuting-input} Channels with commuting inputs:
    Suppose every admissible density operator $\rho_{\A}$ in the input algebra commutes with $\Gamma_{\A\B}^{\mathscr{N}}$, i.e.,
    \begin{align} \label{eq:input_commuting}
        \left[ \rho_{\A} \otimes \I_{\B}, \Gamma_{\A\B}^{\mathscr{N}} \right] = 0_{\A\B},
    \end{align}
    we have 
    \begin{align} \label{eq:commuting-input-expression}
    I_{\alpha}(\mathscr{N}) = \sup_{ \rho_{\A} }
	\frac{\alpha}{\alpha-1} \log \Tr\left[ \left( \Tr_{\A} \left[ \rho_{\A} \cdot \Gamma_{\A\B}^{\alpha}  \right] \right)^{1/\alpha} \right], \quad \alpha \in (0,1).
    \end{align}
\end{enumerate}
\end{app-proposition}
\begin{app-remark}
    A trivial class of channels satisfying \eqref{eq:input_commuting} is the replacer channel $\mathscr{N}_{\A\to\B}(\rho_{\A}) = \sigma_{\B}$ for all $\rho_{\A}$ on $\A$, which has $I_{\alpha}(\mathscr{N}) = 0$.
    Another class of channels satisfying \eqref{eq:input_commuting} is classical-quantum channels $\mathscr{N}_{\X\to\B}$, whose input states $\rho_{\X}$ are restricted to diagonal matrices.
    The corresponding Petz--\Renyi information (for $\alpha \in (0,2]$) is 
    \begin{align} \label{eq:Petz-Renyi_c-q}
        I_{\alpha}(\mathscr{N}_{\X\to\B}) = \sup_{p_{\X}} \frac{\alpha}{\alpha-1} \log \Tr\left[ \left( \sum_{x\in\X} p_{\X}(x) \mathscr{N}(x)^{\alpha} \right)^{1/\alpha} \right]
        = \min_{\sigma_{\B}} \sup_{x\in\X} D_{\alpha}(\mathscr{N}(x)\Vert\sigma_{\B}),
    \end{align}
    where the last term is called \emph{\Renyi divergence radius} and was  proved by Mosonyi and Ogawa \cite[Proposition 4.2]{MO17} (see also \cite{MH11, CGH18}).
    The subadditivity of $I_{\alpha}(\mathscr{N}_{\X\to\B})$ also follows from the min-max expression and the product structure of classical-quantum channels, i.e., $\mathscr{N}_{\X\to\B}^{\otimes 2}(x_1 x_2) = \mathscr{N}_{\X\to\B}(x_1)\otimes\mathscr{N}_{\X\to\B}(x_2)$.
    See also the proof by Li and Yang \cite{LY25}.
\end{app-remark}

\begin{proof}
\noindent Item~\ref{app-item:isometric} (isometric channels):
For an isometry $V_{\A\to\B} |a\rangle_{\A} = |v_{a}\rangle_{\B}$, the joint state is
$\sqrt{\rho_{\A}} \Gamma_{\A\B}^{\mathscr{N}} \sqrt{\rho_{\A}} = |\Psi\rangle\langle\Psi|_{\A\B}$, where
$|\Psi\rangle_{\A\B} = \sum_a \sqrt{\lambda_a} |a\rangle_{\A}|v_a\rangle_{\B}$.
Its $\alpha$-power collapses to itself.
We calculate
\begin{align}
 \left( \Tr_{\A} \left[ \rho_{\A}^{1-\alpha} \left( \sqrt{\rho_{\A}} \Gamma_{\A\B}^{\mathscr{N}} \sqrt{\rho_{\A}} \right)^{\alpha}  \right] \right)^{1/\alpha}
 &= \left( \sum_{a} \lambda_a^{1-\alpha} \lambda_a \vert{}v_a\rangle\langle v_a\vert{}_{\B} \right)^{1/\alpha}
 = \sum_{a} \lambda_a^\frac{2-\alpha}{\alpha} \vert{}v_a\rangle\langle v_a\vert{}_{\B}.
\end{align}

\noindent Item~\ref{app-item:projective} (projective measurement channels):
For any quantum-classical channel $\mathscr{M}_{\A\to\Y}(\rho_{\A}) = \sum_{y\in\Y} \Tr[\rho_{\A}M_{\A}^y] |y\rangle\langle y|_{\Y}$, we calculate the joint state as
\begin{align}
    \omega_{\A\Y}^{\rho}=
	\sum_{y\in\Y} \sqrt{\rho_{\A}} (M_{\A}^y)^{\top} \sqrt{\rho_{\A}}  \otimes \vert{}y\rangle\langle y\vert{}_{\Y}.
\end{align}
Hence, 
\begin{align} 
	I_{\alpha}(\mathscr{M}_{\A\to\Y})
	&= \max_{\rho_{\A}} \frac{\alpha}{\alpha-1} \log \sum_{y\in\Y} \left(\Tr\left[ \rho_{\A}^{1-\alpha} \left(\sqrt{\rho_{\A}} (M_{\A}^y)^{\top}  \sqrt{\rho_{\A}}\right)^{\alpha} \right]\right)^{1/\alpha}
	\\
	&= \max_{\rho_{\A}} \frac{\alpha}{\alpha-1} \log \sum_{y\in\Y} \left(\Tr\left[ \rho_{\A}^{1-\alpha} \left(\sqrt{\rho_{\A}} M_{\A}^y\sqrt{\rho_{\A}}\right)^{\alpha} \right]\right)^{1/\alpha}.
	\label{eq:q-c_formula}
\end{align}
Here, we drop `$\top$' after optimization because $\rho_{\A}\mapsto\rho_{\A}^{\top}$ is a bijection of the state space.

To derive the upper bound on \eqref{eq:q-c_formula} for projective measurements, we invoke the Araki-type trace inequality of Liu--Cheng \cite{LC_Araki}:
\begin{align} \label{eq:Liu-Cheng}
	\Tr\left[f(A) A^s B^s \right]
	\leq \Tr\left[ f(A) \left(A^{\sfrac12} B A^{\sfrac12} \right)^s \right] & s\in(0,1]
\end{align}
with monotone function $f(x) = x^{1-\alpha}$, $s = \alpha \in (0,1)$, $A\leftarrow \rho_{\A}$, and  $B\leftarrow M_{\A}^y$
to obtain
\begin{align}
	I_{\alpha}(\mathscr{M}_{\A\to\Y})
	&\leq \max_{\rho_{\A}} \frac{\alpha}{\alpha-1} \log \sum_{y\in\Y} \left( \Tr\left[ \rho_{\A} \left(M_{\A}^y\right)^{\alpha} \right] \right)^{1/\alpha}
	\\
	&\overset{\text{(a)}}{=} \max_{\rho_{\A}} \frac{\alpha}{\alpha-1} \log \sum_{y\in\Y} \left( \Tr\left[ \rho_{\A} M_{\A}^y \right] \right)^{1/\alpha}
	\\
	&= \max_{\rho_{\A}} H_{1/\alpha}(\Y)_{\sum_{y}\Tr[\rho_{\A}M_{\A}^y]}
	\\
	&\overset{\text{(b)}}{=} \log |\Y|, \quad \alpha \in (0,1).
\end{align}
where (a) follows from the projective measurement and (b) follows from the dimension bound for \Renyi entropies.
For $\alpha \in (1,2]$, we apply \eqref{eq:Liu-Cheng} again with $f(x) = x$, $s = \alpha -1 \in (0,1]$, $A \leftarrow \sqrt{\rho_{\A}} M_{\A}^y \sqrt{\rho_{\A}}$, and $B\leftarrow \rho_{\A}^{-1}$ to obtain the same upper bound on
$I_{\alpha}(\mathscr{M}_{\A\to\Y})$.

The lower bound $I_{\alpha}(\mathscr{M}_{\A\to\Y}) \geq \log |\Y|$ is achieved by choosing $\rho_{\A} = \frac{1}{|\Y|} \sum_{y\in\Y} |u_y\rangle\langle u_y|_{\A}$ satisfying $M_{\A}^y |u_y\rangle_{\A} = |u_y\rangle_{\A}$.

\noindent Item~\ref{app-item:covariant} (covariant channels):
The inner trace function in \eqref{eq:app-objective} is unitary invariant with respect to the underlying group $G$.
Recalling the convexity (resp.~concavity) for $\alpha \in (0,1)$ (resp.~$\alpha \in (1,2]$), the optimizer is attained at the depolarized completely mixed state $\rho_{A}^{\star} = \int_G U_{\A} \rho_{\A} U_{\A}^{\dagger} \, \mathrm{d} U_{\A} = \frac{1}{d_{\A}} \I_{\A}$.
Direct calculation proves the claim.

\noindent Item~\ref{app-item:commuting-input} (commuting input optimizers):
The expression \eqref{eq:commuting-input-expression} directly follows from \eqref{eq:Petz-Renyi_app} and the commutation relation \eqref{eq:input_commuting}.

The superadditivity directly follows from the definition \eqref{eq:Petz-Renyi_app} by choosing the product of optimal marginal states, i.e.~$\rho_{\A_1\A_2} = \rho_{\A_1}^\star \otimes \rho_{\A_2}^\star$.
Below we prove the subadditivity.

Expressing \eqref{eq:commuting-input-expression} in terms of the Schatten norm, we have
\begin{align}
    {I}_{\alpha}(\mathscr{N})
    &= \frac{1}{\alpha-1} \log \inf_{\rho_{\A} \in \mathcal{S}(\A)} \left\| \Tr_{\A}\left[ \rho_{\A} (\Gamma_{\A\B}^{\mathscr{N}})^{\alpha} \right] \right\|_{\frac{1}{\alpha}}
    \\
    &\overset{\text{(a)}}{=} \frac{1}{\alpha-1} \log \inf_{\rho_{\A} \in \mathcal{S}(\A)} \sup_{Z_{\B}\geq 0, \|Z_{\B}\|_{\frac{1}{1-\alpha}} \leq 1 } \Tr\left[ Z_{\B} \Tr_{\A}\left[ \rho_{\A} (\Gamma_{\A\B}^{\mathscr{N}})^{\alpha} \right] \right]
    \\
    &\overset{\text{(b)}}{=} \frac{1}{\alpha-1} \log  \sup_{Z_{\B}\geq 0, \|Z_{\B}\|_{\frac{1}{1-\alpha}} \leq 1 } \inf_{\rho_{\A} \in \mathcal{S}(\A)}\Tr\left[ Z_{\B} \Tr_{\A}\left[ \rho_{\A} (\Gamma_{\A\B}^{\mathscr{N}})^{\alpha} \right] \right]
    \\
    &= \frac{1}{\alpha-1} \log  \sup_{Z_{\B}\geq 0, \|Z_{\B}\|_{\frac{1}{1-\alpha}} \leq 1 } \inf_{\rho_{\A} \in \mathcal{S}(\A)}\Tr\left[ \rho_{\A} \otimes  Z_{\B} (\Gamma_{\A\B}^{\mathscr{N}})^{\alpha} \right]
    \\
    &= \frac{1}{\alpha-1} \log  \sup_{Z_{\B}\geq 0, \|Z_{\B}\|_{\frac{1}{1-\alpha}} \leq 1 } \lambda_{\min} \left( \Tr_{\B}\left[ Z_{\B} (\Gamma_{\A\B}^{\mathscr{N}})^{\alpha} \right] \right),
    \label{eq:upper_additive1}
\end{align}
(a) follows from the duality of Schatten $\frac{1}{\alpha}$ and $\frac{1}{1-\alpha}$ norms;
(b) follows from Sion's minimax theorem and the bilinearity of the  objective function.
The subadditivity then follows by choosing $Z_{\B_1\B_2} = Z_{\B_1}^\star \otimes Z_{\B_2}^\star$, where $Z_{\B_i}$ is the optimizer in \eqref{eq:upper_additive1} for $\mathscr{N}_i$.
This concludes the proof.
\end{proof}

\begin{app-proposition}[Additivity at $\alpha =2$] \label{prop:alpha-2}
The Petz--\Renyi information $I_{\alpha}(\mathscr{N})$ is strongly additive for $\alpha = 2$, i.e.,
\begin{align}
    I_2(\mathscr{N}_1\otimes\mathscr{N}_2)
    = I_2(\mathscr{N}_1) + I_2(\mathscr{N}_2).
\end{align}
\end{app-proposition}

\begin{proof}
It suffices to prove the subadditivity since the superadditivity follows directly from the definition.
Define
\begin{align}
 Q_2(\mathscr{N})
 &\coloneqq
 \sup_{\rho_{\A}}
 \Tr\!\left[
 \left(
 \Tr_{\A}\!\left[
 \rho_{\A}^{-1}
 \left(
 \sqrt{\rho_{\A}}\,
 \Gamma_{\A\B}^{\mathscr{N}}
 \sqrt{\rho_{\A}}
 \right)^2
 \right]
 \right)^{1/2}
 \right].
 \label{eq:def-Q}
\end{align}
Hence, $ I_2(\mathscr{N})=2\log Q_2(\mathscr{N})$.
Let $\Pi_{\rho_{\A}}$ be the projection onto the support of $\rho_{\A}$. 
By the cyclic property of trace, we have, 
\begin{align}
 Q_2(\mathscr{N})
 &= \sup_{\rho_{\A}}
 \Tr\!\sqrt{
 \Tr_{\A}\!\left[ \Pi_{\rho_{\A}}
 \Gamma_{\A\B}^{\mathscr{N}}
 \rho_{\A}
 \Gamma_{\A\B}^{\mathscr{N}} \Pi_{\rho_{\A}}
 \right]}
 \\
 &\overset{\text{(a)}}{=}\sup_{\rho_{\A}>0, \Tr[\rho_{\A}]=1}
 \Tr\!\sqrt{
 \Tr_{\A}\!\left[
 \Gamma_{\A\B}^{\mathscr{N}}
 \rho_{\A}
 \Gamma_{\A\B}^{\mathscr{N}}
 \right]}
 \\
 &=\max_{\rho_{\A}}
 \Tr\!\sqrt{
 \Tr_{\A}\!\left[
 \Gamma_{\A\B}^{\mathscr{N}}
 \rho_{\A}
 \Gamma_{\A\B}^{\mathscr{N}}
 \right]},
 \label{eq:primal-Q}
\end{align}
where in (a) we restricted the optimization to the dense subset of full-rank $\rho_{\A}$ because the objective function here is continuous in $\rho_{\A}$.

Recall the Cauchy--Schwartz inequality, we have, for every $X_{\B}\geq0$,
\begin{align}
 \left(\Tr\sqrt{X_{\B}}\right)^2
 =\inf_{\sigma_{\B} > 0,\, \Tr[\sigma_{\B}]=1}
 \Tr[X_{\B}\sigma_{\B}^{-1}].
 \label{eq:sqrt-variational}
\end{align}
Hence,
\begin{align}
    Q_2(\mathscr{N})^2
    &= \max_{\rho_{\A}} \inf_{\sigma_{\B} > 0,\, \Tr[\sigma_{\B}]=1} \Tr\!\left[ \rho_{\A}
 \Gamma_{\A\B}^{\mathscr{N}}
 (\I_{\A}\otimes\sigma_{\B}^{-1})
 \Gamma_{\A\B}^{\mathscr{N}}
 \right]
 \\
 &= \inf_{\sigma_{\B} > 0,\, \Tr[\sigma_{\B}]=1} \max_{\rho_{\A}} \Tr\!\left[ \rho_{\A}
 \Gamma_{\A\B}^{\mathscr{N}}
 (\I_{\A}\otimes\sigma_{\B}^{-1})
 \Gamma_{\A\B}^{\mathscr{N}}
 \right]
 \\
 &= \inf_{\sigma_{\B} > 0,\, \Tr[\sigma_{\B}]=1} \left\| \Tr_{\B}\!\left[
 \Gamma_{\A\B}^{\mathscr{N}}
 (\I_{\A}\otimes\sigma_{\B}^{-1})
 \Gamma_{\A\B}^{\mathscr{N}}
 \right] \right\|_{\infty}.
\end{align}
Here, the objective function $(\rho_{\A},\sigma_{\B}) \mapsto \Tr_{\B}\!\left[ \rho_{\A}
 \Gamma_{\A\B}^{\mathscr{N}}
 (\I_{\A}\otimes\sigma_{\B}^{-1})
 \Gamma_{\A\B}^{\mathscr{N}}
 \right]$ is linear and convex via the operator convexity of the inverse function.
We apply Sion's minimax theorem by imposing $\sigma_{\B} \geq \epsilon \I_{\B}$ for the  resulting compact convex set and then letting $\epsilon \searrow 0$.

The multiplicativity of $Q_2(\mathscr{N})$ then follows by choosing product state $\sigma_{\B_1\B_2} = \sigma_{\B_1}^\star \otimes \sigma_{\B_2}^\star$ and the multiplicativity of the operator norm $\|\,\cdot\,\|_\infty$.
\end{proof}

\section{Analytic Strict Superadditivity: Fourier Measurements} \label{sec:q-c}

This section establishes strict superadditivity for a Fourier measurement (Theorem~\ref{thm:strict-superadditivity-app}).
The one-copy optimization reduces to a diagonal one-dimensional convex problem.  
The strict two-copy improvement is then proved analytically by a linear-response
calculation around the product of the one-copy optimizer toward a diagonal movement.
Notably, the witness of the strict superadditivity is only via a classically correlated state on the two-copy system.

\subsection{The single-heavy Fourier measurement}

Throughout this section, we write $d = d_{\A}$, and shorthand $M_y = M_{\A}^y$.
Fix
\begin{align}
  d\ge 3,
  \qquad
  \frac1d<\lambda<1,
  \qquad
  b=\frac{1-\lambda}{d-1}.
\end{align}
Let
\begin{align}
  a_0=\lambda,
  \qquad
  a_1=\cdots=a_{d-1}=b,
\end{align}
and let \(\omega=\mathrm{e}^{2\pi \mathrm{i}/d}\).  On \(\A\cong\mathbb{C}^d\), define
\begin{align}
  |v_y\rangle
  =
  \sum_{j=0}^{d-1}\sqrt{a_j}\,\omega^{jy}|j\rangle,
  \qquad
  y=0,\ldots,d-1.
\end{align}
The measurement \(\mathscr{M}\) consists of the \(d\) rank-one effects and the residual effect:
\begin{subequations} \label{eq:Fourier_POVM}
\begin{align} 
  M_y&=\frac{1}{d\lambda} |v_y\rangle\!\langle v_y|,
  \qquad
  y=0,\ldots,d-1,
  \\
  M_\infty
  &=
  \left(1-\frac{b}{\lambda}\right) \sum_{j=1}^{d-1}|j\rangle\!\langle j|.
\end{align}
\end{subequations}
Indeed,
\begin{align}
  \sum_{y=0}^{d-1}M_y
  =
  \sum_{j=0}^{d-1}\frac{a_j}{\lambda}|j\rangle\!\langle j|,
\end{align}
and therefore \(\sum_{y=0}^{d-1}M_y+M_\infty=\I\).

\begin{figure}
\centering
\begin{tikzpicture}[x=1cm, y=1cm, scale=1.1, transform shape]

    \definecolor{planeBg}{RGB}{247, 248, 250}       
    \definecolor{planeGrid}{RGB}{225, 228, 235}     
    \definecolor{planeBorder}{RGB}{170, 180, 195}   
    \definecolor{axisLine}{RGB}{45, 55, 70}         
    \definecolor{vecM}{RGB}{210, 50, 50}            
    \definecolor{vecProj}{RGB}{140, 150, 165}       
    \definecolor{coneFill}{RGB}{220, 80, 80}        

    \coordinate (Origin) at (0,0);
    
    \coordinate (B0) at ({2.5*cos(-45)}, {0.85*sin(-45)});   
    \coordinate (B1) at ({2.5*cos(75)}, {0.85*sin(75)});     
    \coordinate (B2) at ({2.5*cos(195)}, {0.85*sin(195)});   

    \coordinate (V0) at ({2.5*cos(-45)}, {2.6 + 0.85*sin(-45)});
    \coordinate (V1) at ({2.5*cos(75)}, {2.6 + 0.85*sin(75)});
    \coordinate (V2) at ({2.5*cos(195)}, {2.6 + 0.85*sin(195)});
    
    \coordinate (Center) at (0, 2.6);

    \tikzset{vector arrow/.style={->, >={Stealth[length=3mm, width=1.8mm]}, color=vecM, line width=1.8pt}}

    \filldraw[fill=planeBg, draw=planeBorder, line width=1pt] (0,0) ellipse (3.6cm and 1.25cm);
    
    \begin{scope}
        \clip (0,0) ellipse (3.6cm and 1.25cm);
        \draw[color=planeGrid, line width=0.6pt, step=0.4cm] (-4,-1.5) grid (4,1.5);
    \end{scope}

    \node[color=axisLine, font=\small, anchor=north west] at (-1.8, -0.6) {Residual Subspace $M_\infty$};

    \draw[color=vecProj, line width=0.8pt, densely dashed] (B0) -- (B1) -- (B2) -- cycle;
    \draw[color=vecProj, line width=0.8pt, densely dashed] (Origin) -- (B0);
    \draw[color=vecProj, line width=0.8pt, densely dashed] (Origin) -- (B1);
    \draw[color=vecProj, line width=0.8pt, densely dashed] (Origin) -- (B2);
    
    \draw[color=vecProj, line width=0.6pt, densely dotted] (V0) -- (B0);
    \draw[color=vecProj, line width=0.6pt, densely dotted] (V1) -- (B1);
    \draw[color=vecProj, line width=0.6pt, densely dotted] (V2) -- (B2);

    \node[color=vecProj, font=\sffamily\normalsize, fill=planeBg, inner sep=1pt] at (0.9, -0.3) {$\sqrt{b}$};

    \draw[vector arrow] (Origin) -- (V1) node[above right, font=\sffamily\large] {$M_1$};

    \fill[coneFill, opacity=0.10] (Origin) -- (V2) -- (V1) -- cycle;
    \fill[coneFill, opacity=0.10] (Origin) -- (V0) -- (V1) -- cycle;

    \draw[->, >={Stealth[length=3mm, width=1.8mm]}, color=axisLine, line width=1.5pt] 
        (Origin) -- (0, 3.8) node[right, font=\sffamily\large, yshift=-0.1cm] {$|0\rangle$};

    \draw[color=axisLine, line width=0.8pt] (-0.1, 2.6) -- (0.1, 2.6);
    \node[color=axisLine, font=\sffamily\normalsize, anchor=east] at (-0.15, 2.6) {$\sqrt{\lambda}$};

    \fill[coneFill, opacity=0.15] (Origin) -- (V2) -- (V0) -- cycle;

    \draw[vector arrow] (Origin) -- (V2) node[above left, font=\sffamily\large] {$M_2$};
    \draw[vector arrow] (Origin) -- (V0) node[below right, font=\sffamily\large] {$M_0$};

\end{tikzpicture}
\end{figure}

For notational simplicity, we define the one-copy functional
\begin{align}
  Q_\alpha(\rho;\mathscr{M})
  =
  \sum_y
  \left(
  \Tr\left[
  \rho^{1-\alpha}
  \left(\rho^{1/2}M_y^{\top}\rho^{1/2}\right)^\alpha
  \right]
  \right)^{1/\alpha}, \quad \alpha \in (0,1),
\end{align}
where the sum includes the residual outcome \(y=\infty\), and define
\begin{align} \label{eq:K_POVM}
    K^{(n)}(\alpha)
    \equiv \inf_{\rho_{\A^n}} Q_{\alpha}(\rho_{\A^n}; \mathscr{M}^{\otimes n}).
\end{align}
The strict superadditivity
\begin{align}
    I_{\alpha}(\mathscr{M}^{\otimes 2}) > 2 I_{\alpha}(\mathscr{M})
\end{align}
is equivalent to
\begin{align}
  K^{(2)}(\alpha)<K^{(1)}(\alpha)^2.
\end{align}

\subsection{Diagonal reduction of the one-copy optimization}

The diagonal reduction follows from covariance and convexity.  
Let \(Z = \sum_{j=0}^{d-1} \omega^j |j\rangle\langle j|\) be the generalized Pauli phase operator and 
\begin{align}
  Z^m|j\rangle=\omega^{mj}|j\rangle,
  \qquad
  m=0,\ldots,d-1.
\end{align}
The phase operator permutes the transposed Fourier effects \(M_y^{\top}\) and leaves
\(M_\infty^{\top}=M_\infty\) fixed.  Hence
\begin{align}
  Q_\alpha(Z^m\rho Z^{-m};\mathscr{M})=Q_\alpha(\rho;\mathscr{M}).
\end{align}
The phase twirl
\begin{align}
  \Delta(\rho)
  =
  \frac1d\sum_{m=0}^{d-1}Z^m \rho Z^{-m}
  =
  \sum_{j=0}^{d-1}|j\rangle\langle j|\rho|j\rangle\langle j|
\end{align}
therefore satisfies, by the convexity proven in Proposition~\ref{prop:Petz_convexity},
\begin{align}
  Q_\alpha(\Delta(\rho);\mathscr{M})
  \le
  \frac1d\sum_{m=0}^{d-1}Q_\alpha(Z^m \rho Z^{-m};\mathscr{M})
  =
  Q_\alpha(\rho;\mathscr{M}).
\end{align}
Hence, a one-copy optimizer may be chosen diagonal.

For a diagonal state \(p=(p_0,\ldots,p_{d-1})\), the objective is
\begin{equation}
\label{eq:diagonal-objective}
  Q_\alpha(p;\mathscr{M})
  =
  \frac{1}{\lambda}
  \left(\sum_{j=0}^{d-1}a_jp_j\right)^{\frac{\alpha-1}{\alpha}}
  \left(\sum_{j=0}^{d-1}a_jp_j^{2-\alpha}\right)^{\frac1\alpha}
  +
  \left(1 - \frac{b}{\lambda}\right)\left(\sum_{j=1}^{d-1}p_j\right)^{1/\alpha}.
\end{equation}
Fix \(t=p_0\).  The first factor \(\sum_j a_jp_j=\lambda t+b(1-t)\) and the residual
term depend only on \(t\).  The remaining tail dependence is through
\(\sum_{j=1}^{d-1}p_j^{2-\alpha}\).  Since \(2-\alpha>1\), the function
\(r\mapsto r^{2-\alpha}\) is strictly convex.  Hence, for fixed tail mass \(1-t\),
this sum is uniquely minimized by the uniform tail
\begin{align}
	p_1=\cdots=p_{d-1}=\frac{1-t}{d-1}.
\end{align}
Therefore
\begin{align}
  K^{(1)}(\alpha)=\min_{0\le t\le1} q_\alpha(t),
\end{align}
where
\begin{align}
  q_\alpha(t)&=r_\alpha(t)+z_\alpha(t),
  \\
  r_\alpha(t)
  &=
  \frac{1}{\lambda}\,\eta(t)^{\frac{\alpha-1}{\alpha}}\,
  \zeta_\alpha(t)^{\frac1\alpha},
  \\
  z_\alpha(t)
  &=
  \left(1-\frac{b}{\lambda}\right)(1-t)^{1/\alpha},
\end{align}
and
\begin{align}
  \eta(t)=\lambda t+b(1-t),
  \qquad
  \zeta_\alpha(t)
  =
  \lambda t^{2-\alpha}+(d-1)b\left(\frac{1-t}{d-1}\right)^{2-\alpha}.
\end{align}
This is the diagonal one-parameter reduction.

\begin{app-lemma}[Convexity of the reduced one-copy objective, POVM]
	For every \(0<\alpha<1\), the function
	\begin{align}
		t\mapsto q_{\alpha}(t)
		=
		Q_\alpha(\rho_t;\mathscr{M})
	\end{align}
	is convex on \([0,1]\).
\end{app-lemma}

\begin{proof}
	Set \(p=1/\alpha>1\) and define
	\begin{align}
		\phi(z,e)
		:=
		e\left(\frac{z}{e}\right)^p
		=
		z^p e^{1-p},
		\qquad
		z\geq0,\quad e>0.
	\end{align}
	The function \(\phi\) is the perspective of the convex function
	\(z\mapsto z^p\), and hence is jointly convex. Moreover,
	\begin{align}
		\frac{\partial\phi}{\partial z}(z,e)
		=
		pz^{p-1}e^{1-p}
		\geq0,
	\end{align}
	thereby \(\phi\) is nondecreasing in its first argument.

	Notice that \(\eta(t)>0\) is affine and
	\(\zeta_\alpha(t)\) is convex on \([0,1]\), since
	\(2-\alpha>1\). Therefore, for every \(s,t\in[0,1]\) and
	\(\theta\in[0,1]\),
	\begin{align}
		&r_\alpha\bigl(\theta t+(1-\theta)s\bigr)
		\nonumber\\
		&\quad=
		\frac{1}{\lambda}
		\phi\left(
			\zeta_\alpha\bigl(\theta t+(1-\theta)s\bigr),
			\eta\bigl(\theta t+(1-\theta)s\bigr)
		\right)
		\nonumber\\
		&\quad\leq
		\frac{1}{\lambda}
		\phi\left(
			\theta\zeta_\alpha(t)+(1-\theta)\zeta_\alpha(s),
			\theta\eta(t)+(1-\theta)\eta(s)
		\right)
		\nonumber\\
		&\quad\leq
		\theta r_\alpha(t)+(1-\theta)r_\alpha(s).
	\end{align}
	The first inequality follows from the convexity of
	\(\zeta_\alpha\), the affinity of \(\eta\), and the monotonicity
	of \(\phi\) in its first argument; the second follows from the
	joint convexity of \(\phi\). Hence \(r_\alpha\) is convex.

	Finally, since \(1/\alpha>1\),
	\begin{align}
		z_\alpha(t)
		=
		\left(1-\frac{b}{\lambda}\right)(1-t)^{1/\alpha}
	\end{align}
	is convex. Therefore
	\(q_\alpha(t)=r_\alpha(t)+z_\alpha(t)\) is convex on \([0,1]\).
\end{proof}

Next, we show that the minimizer is an interior point.  Indeed, a direct endpoint check gives
\begin{align}
  q_\alpha'(0^+)<0,
  \qquad
  q_\alpha'(1^-)>0.
\end{align}
Let \(t_\alpha\in(0,1)\) be a global minimizer.  Then
\begin{equation}
\label{eq:stationarity}
  q_\alpha'(t_\alpha)=0.
\end{equation}
Moreover, at the uniform point \(t=1/d\),
\begin{equation}
\label{eq:uniform-derivative-new}
  q_\alpha'\!\left(\frac1d\right)
  =
  \frac{d(\lambda-b)}{\alpha \lambda}\,d^{-1/\alpha}
  \left[1-(d-1)^{1/\alpha-1}\right]
  <0,
\end{equation}
where the strict inequality uses \(d\ge3\) and \(\alpha<1\).  Hence no global
minimizer can be the uniform point:
\begin{equation}
\label{eq:not-uniform}
  t_\alpha\ne \frac1d.
\end{equation}

\subsection{A correlated diagonal two-copy path}

Let
\begin{align}
  p=p(t_\alpha)
  =
  \left(t_\alpha,\frac{1-t_\alpha}{d-1},\ldots,\frac{1-t_\alpha}{d-1}\right),
\end{align}
and define the zero-sum vector
\begin{align}
  u=\left(1,-\frac1{d-1},\ldots,-\frac1{d-1}\right).
\end{align}
For real \(\kappa\), set
\begin{equation}
\label{eq:pi-kappa}
  \pi_\kappa(i,j)=p_ip_j+\kappa \cdot u_i u_j.
\end{equation}
For all sufficiently small \(|\kappa|\), \(\pi_\kappa\) is a bipartite probability distribution.
Because \(\sum_i u_i=0\), it has the same one-copy marginals as \(p\).
Denote a diagonal state on $\mathbb{C}^{d\times d}$ by
\begin{align} \label{eq:two-copy_POVM}
  \rho_{\alpha}(\kappa)
  =
  \sum_{i,j=0}^{d-1}\pi_\kappa(i,j)|ij\rangle\!\langle ij|.
\end{align}
At \(\kappa=0\), this is the product of the one-copy optimizer with itself, i.e., \(\pi_0 = p\otimes p\).
Hence,
\begin{align}
  Q_\alpha(\rho_{\alpha}(0);\mathscr{M}^{\otimes 2})=K^{(1)}(\alpha)^2.
\end{align}
Define the linear response
\begin{align}
  c_1(\alpha)
  =
  \left.
  \frac{\mathrm{d}}{\mathrm{d}\kappa}Q_\alpha(\rho_{\alpha}(\kappa);\mathscr{M}^{\otimes 2})
  \right|_{\kappa=0}.
\end{align}
If \(c_1(\alpha)<0\), then for sufficiently small positive \(\kappa\),
\begin{align}
  Q_\alpha(\rho_{\alpha}(\kappa);\mathscr{M}^{\otimes 2})
  = K^{(1)}(\alpha)^2 + c_{1}(\alpha)\kappa + \mathcal{O}(\kappa^2)
  <K^{(1)}(\alpha)^2,
\end{align}
and hence \(K^{(2)}(\alpha)<K^{(1)}(\alpha)^2\).

\subsection{Linear response}

For \(0<t<1\), define
\begin{align}
  \xi(t)=\frac{\lambda-b}{\eta(t)}
\end{align}
and
\begin{align}
  \theta_\alpha(t)
  =
  \frac{1}{\zeta_\alpha(t)} \left[ \lambda t^{1-\alpha} - b \left(\frac{1-t}{d-1}\right)^{1-\alpha} \right].
\end{align}
Then
\begin{align}
  \frac{r_\alpha'(t)}{r_\alpha(t)}
  =
  \varphi_\alpha(t),
  \qquad
  \varphi_\alpha(t)
  =
  \frac{\alpha-1}{\alpha}\xi(t)
  +
  \frac{2-\alpha}{\alpha}\theta_\alpha(t),
\end{align}
and
\begin{align}
  z_\alpha'(t)=-\frac{z_\alpha(t)}{\alpha(1-t)}.
\end{align}
The stationarity condition \eqref{eq:stationarity} is therefore
\begin{equation}
\label{eq:rz-stationarity}
  r_\alpha(t_\alpha)\varphi_\alpha(t_\alpha)
  =
  \frac{z_\alpha(t_\alpha)}{\alpha(1-t_\alpha)}.
\end{equation}

\begin{app-lemma}[Linear response at the one-copy minimizer]
\label{lem:linear-response}
At \(t=t_\alpha\),
\begin{equation}
\label{eq:c1-formula}
  c_1(\alpha)
  =
  \frac{(\alpha-1)(2-\alpha)}{\alpha}
  r_\alpha(t_\alpha)^2
  \left[\xi(t_\alpha)-\theta_\alpha(t_\alpha)\right]^2.
\end{equation}
\end{app-lemma}

\begin{proof}
Differentiate the four two-copy outcome types along the path
\eqref{eq:pi-kappa}: Fourier--Fourier, Fourier--residual, residual--Fourier, and
residual--residual.  At \(\kappa=0\), all four contributions factor into the one-copy
terms \(r_\alpha\) and \(z_\alpha\).  A direct differentiation gives
\begin{align*}
  c_1(\alpha)
  &=
  r_\alpha^2
  \left[
  \frac{\alpha-1}{\alpha}\xi^2
  +
  \frac{2-\alpha}{\alpha}\theta_\alpha^2
  \right]
  -
  \frac{2\varphi_\alpha}{1-t}\,r_\alpha z_\alpha
  +
  \frac{z_\alpha^2}{\alpha(1-t)^2},
\end{align*}
where all quantities on the right are evaluated at \(t=t_\alpha\).  Using
\eqref{eq:rz-stationarity}, the last two terms become
\(-\alpha r_\alpha^2\varphi_\alpha^2\).  Hence
\begin{align}
  c_1(\alpha)
  =
  r_\alpha^2
  \left[
  \frac{\alpha-1}{\alpha}\xi^2
  +
  \frac{2-\alpha}{\alpha}\theta_\alpha^2
  -
  \alpha\varphi_\alpha^2
  \right].
\end{align}
Since
\begin{align}
  \alpha\varphi_\alpha=(\alpha-1)\xi+(2-\alpha)\theta_\alpha,
\end{align}
the bracket is
\begin{align}
  \frac{(\alpha-1)(2-\alpha)}{\alpha}(\xi-\theta_\alpha)^2.
\end{align}
This proves \eqref{eq:c1-formula}.
\end{proof}

The prefactor in \eqref{eq:c1-formula} is strictly negative for \(0<\alpha<1\).
It remains only to show that the square is nonzero.  The equality case is explicit:
\begin{align}
  \xi(t)=\theta_\alpha(t)
  \quad\Longleftrightarrow\quad
  t=\frac1d.
\end{align}
Indeed,
\begin{align}
  \left(\lambda t^{1-\alpha}-b(\tfrac{1-t}{d-1})^{1-\alpha}\right)\eta(t)
  -(\lambda-b)\zeta_\alpha(t)
  =
  \lambda b\left(t^{1-\alpha}-(\tfrac{1-t}{d-1})^{1-\alpha}\right).
\end{align}
We have \(\xi(t)=\theta_\alpha(t)\) if and only if \(t^{1-\alpha}=(\frac{1-t}{d-1})^{1-\alpha}\), equivalently \(t=1/d\).  By \eqref{eq:not-uniform}, the minimizer
\(t_\alpha\) is not \(1/d\).  Therefore,
\begin{align}
  \xi(t_\alpha)\ne\theta_\alpha(t_\alpha),
\end{align}
and \eqref{eq:c1-formula} gives
\begin{equation}
\label{eq:c1-negative}
  c_1(\alpha)<0.
\end{equation}

\subsection{Strict superadditivity}

\begin{app-theorem}[Strict superadditivity of Fourier measurements]
\label{thm:strict-superadditivity-app}
For every \(d\ge3\), every
\(\lambda\in(1/d,1)\), and every \(0<\alpha<1\), the single-heavy Fourier measurement defined in \eqref{eq:Fourier_POVM}  satisfies
\begin{align}
  I_{\alpha}(\mathscr{M}^{\otimes 2}) > 2 I_{\alpha}(\mathscr{M}).
\end{align}
The strict improvement is witnessed by the infinitesimal diagonal classically
correlated path \(\rho_\alpha(\kappa)\) in \eqref{eq:two-copy_POVM}.
\end{app-theorem}

\begin{app-remark}
	In this paper, we only focus on the range $\alpha \in (0,1)$ for $I_{\alpha}(\mathscr{N})$.
	However, the proof of Theorem~\ref{thm:strict-superadditivity-app} naturally extends to $\alpha \in (1,2)$.
\end{app-remark}

\begin{proof}
By the diagonal reduction, \(K^{(1)}(\alpha)=q_\alpha(t_\alpha)\).  The product point
\(\rho_\alpha(0)\) satisfies
\begin{align}
  Q_\alpha(\rho_\alpha(0);\mathscr{M}^{\otimes 2})=K^{(1)}(\alpha)^2.
\end{align}
The linear-response calculation gives \(c_1(\alpha)<0\).  Hence, for sufficiently
small positive \(\kappa\),
\begin{align}
  Q_\alpha(\rho_\alpha(\kappa);\mathscr{M}^{\otimes 2})<K^{(1)}(\alpha)^2.
\end{align}
Therefore
\begin{align}
  K^{(2)}(\alpha)<K^{(1)}(\alpha)^2.
\end{align}
Multiplying the logarithmic inequality by the negative number \(\alpha/(\alpha-1)\)
proves the claim.
\end{proof}

\subsection{Quadratic choice of the correlation strength}

For strictness, the sign \(c_1(\alpha)<0\) is enough.  In computations one may choose
\(\kappa\) by the one-dimensional minimization of
\(\kappa\mapsto Q_\alpha(\rho_\kappa;\mathscr{M}^{\otimes 2})\).  Equivalently, if
\begin{align}
  Q_\alpha(\rho_\alpha(\kappa);\mathscr{M}^{\otimes 2})
  =
  K^{(1)}(\alpha)^2+c_1(\alpha)\kappa+c_2(\alpha)\kappa^2+O(\kappa^3)
\end{align}
and \(c_2(\alpha)>0\), the quadratic approximation gives
\begin{align} \label{eq:kappa_quad}
  \kappa_{\mathrm{quad}}(\alpha)
  =
  -\frac{c_1(\alpha)}{2c_2(\alpha)}.
\end{align}
This is the clean analytic version of the optimized-\(\kappa\) witness.  It is
especially useful near \(\alpha=1\), where the final gap becomes very small.  The
proof above does not require resolving that small gap numerically; it only uses the
exact sign of the derivative \(c_1(\alpha)\), which remains negative for every fixed
\(\alpha<1\).

\subsection{The four-dimensional instance}

The following numerical examples illustrate the scale of the linear response with parameters
\begin{align}
d=4,
\qquad
\lambda=\frac9{25},
\qquad
b=\frac{16}{75}.
\end{align}

\begin{center}\small
\begin{tabular}{c c c c}
\toprule
\(\alpha\) & \(c_1(\alpha)\) & \(\kappa_{\mathrm{quad}}(\alpha)\) & leading-order gap \\
\midrule
\(0.50\) & \(-9.818\times10^{-3}\) & \(2.563\times10^{-3}\) & \(9.13\times10^{-5}\) \\
\(0.60\) & \(-6.651\times10^{-3}\) & \(1.510\times10^{-3}\) & \(2.91\times10^{-5}\) \\
\(0.70\) & \(-3.230\times10^{-3}\) & \(7.570\times10^{-4}\) & \(6.91\times10^{-6}\) \\
\(0.80\) & \(-1.021\times10^{-3}\) & \(3.021\times10^{-4}\) & \(1.04\times10^{-6}\) \\
\(0.90\) & \(-1.308\times10^{-4}\) & \(6.968\times10^{-5}\) & \(5.19\times10^{-8}\) \\
\bottomrule
\end{tabular}
\end{center}

For $\alpha \in [\sfrac12,1)$ and the above chosen parameters, one can also choose an explicit $\kappa(\alpha)$ as
\begin{align} \label{eq:kappa_closed}
	\kappa(\alpha) 
	= \frac{7}{250}(1-\alpha)^2 (2-\alpha) t_{\alpha} (1-t_{\alpha}).
\end{align}
Figure~\ref{fig:Fourier_POVM} plots the above numeric example.
However, note that \eqref{eq:kappa_closed} is not a universal witness; one has to consider at least the second-order derivative bound.

The gap tends to zero as \(\alpha\nearrow1\), but strictness does not rely on a
floating-point comparison of nearly equal numbers.  The analytic certificate is the
identity \eqref{eq:c1-formula} together with \(t_\alpha\ne1/d\).

\section{Analytic Strict Superadditivity: Amplitude Damping}
\label{sec:GAD}

Let \(\mathscr N_{\A\to\B}\) be the amplitude-damping channel with
damping parameter \(0<\gamma<1\), and put
\begin{align}
    \eta:=1-\gamma.
\end{align}
Its Choi matrix is
\begin{align}
    \Gamma_{\A\B}^{\mathscr N}
    =
    \begin{pmatrix}
        1&0&0&\sqrt{\eta}\\
        0&0&0&0\\
        0&0&\gamma&0\\
        \sqrt{\eta}&0&0&\eta
    \end{pmatrix}.
    \label{eq:GAD-app}
\end{align}
Equivalently,
\begin{align}
    \Gamma_{\A\B}^{\mathscr N}
    =
    |\phi_0\rangle\!\langle\phi_0|
    +
    |\phi_1\rangle\!\langle\phi_1|,
    \qquad
    |\phi_0\rangle
    =
    |00\rangle+\sqrt{\eta}\,|11\rangle,
    \qquad
    |\phi_1\rangle
    =
    \sqrt{\gamma}\,|10\rangle.
    \label{eq:AD-Choi-branches}
\end{align}

For \(0<\alpha<1\), recall the Petz--Rényi trace functional
\begin{align}
    Q_{\alpha}(\rho;\mathscr N)
    &:=
    \Tr\left[
        \left(
            \Tr_{\A}\left[
                \rho_{\A}^{1-\alpha}
                \left(
                    \sqrt{\rho_{\A}}\,
                    \Gamma_{\A\B}^{\mathscr N}\,
                    \sqrt{\rho_{\A}}
                \right)^{\alpha}
            \right]
        \right)^{1/\alpha}
    \right],
    \label{eq:AD-Q-def}
    \\
    K^{(n)}(\alpha)
    &:=
    \inf_{\rho_{\A^n}\in\mathcal D(\A^n)}
    Q_{\alpha}(\rho_{\A^n};\mathscr N^{\otimes n}).
    \label{eq:AD-K-def}
\end{align}
Since \(0<\alpha<1\), the coefficient
\(\alpha/(\alpha-1)\) is negative, and hence
\begin{align}
    I_{\alpha}(\mathscr N^{\otimes n})
    =
    \frac{\alpha}{\alpha-1}\log K^{(n)}(\alpha).
    \label{eq:AD-I-via-K}
\end{align}
Consequently, it is enough to prove
\begin{align}
    K^{(2)}(\alpha)
    <
    \bigl(K^{(1)}(\alpha)\bigr)^2.
    \label{eq:AD-target-K}
\end{align}

\subsection{Diagonal one-copy reduction}

Let
\begin{align}
    Z:=|0\rangle\!\langle0|-|1\rangle\!\langle1|.
\end{align}
The amplitude-damping Choi matrix satisfies
\begin{align}
    (Z_{\A}\otimes Z_{\B})
    \Gamma_{\A\B}^{\mathscr N}
    (Z_{\A}\otimes Z_{\B})
    =
    \Gamma_{\A\B}^{\mathscr N}.
    \label{eq:AD-Z-Choi}
\end{align}

For convenience, define the positive operator
\begin{align}
    T_{\alpha}^{\mathscr N}(\rho)
    :=
    \Tr_{\A}\left[
        \rho_{\A}^{1-\alpha}
        \left(
            \sqrt{\rho_{\A}}\,
            \Gamma_{\A\B}^{\mathscr N}\,
            \sqrt{\rho_{\A}}
        \right)^{\alpha}
    \right],
    \label{eq:AD-T-def}
\end{align}
so that
\begin{align}
    Q_{\alpha}(\rho;\mathscr N)
    =
    \Tr\left[
        T_{\alpha}^{\mathscr N}(\rho)^{1/\alpha}
    \right].
\end{align}

\begin{app-lemma}[Diagonal reduction for amplitude damping]
\label{lemm:AD-diagonal-reduction}
For every \(0<\alpha<1\),
\begin{align}
    K^{(1)}(\alpha)
    =
    \min_{0\leq t\leq1}
    Q_{\alpha}(\rho_t;\mathscr N),
    \qquad
    \rho_t:=\mathrm{diag}(t,1-t).
    \label{eq:AD-diagonal-reduction}
\end{align}
\end{app-lemma}

\begin{proof}
For every state \(\rho\), \eqref{eq:AD-Z-Choi} gives
\begin{align}
    \sqrt{Z\rho Z}\,
    \Gamma^{\mathscr N}\,
    \sqrt{Z\rho Z}
    =
    (Z_{\A}\otimes Z_{\B})
    \left(
        \sqrt{\rho}\,
        \Gamma^{\mathscr N}\,
        \sqrt{\rho}
    \right)
    (Z_{\A}\otimes Z_{\B}).
\end{align}
Using also
\((Z\rho Z)^{1-\alpha}=Z\rho^{1-\alpha}Z\), we obtain
\begin{align}
    T_{\alpha}^{\mathscr N}(Z\rho Z)
    =
    Z_{\B}T_{\alpha}^{\mathscr N}(\rho)Z_{\B}.
\end{align}
Therefore
\begin{align}
    Q_{\alpha}(Z\rho Z;\mathscr N)
    =
    Q_{\alpha}(\rho;\mathscr N).
    \label{eq:AD-Z-invariance}
\end{align}

Consider the computational-basis pinching
\begin{align}
    \Delta_Z(\rho)
    :=
    \frac12\left(\rho+Z\rho Z\right).
\end{align}
For a qubit, \(\Delta_Z(\rho)\) is diagonal. By
Proposition~\ref{prop:Petz_convexity}, the map
\(\rho\mapsto Q_{\alpha}(\rho;\mathscr N)\) is convex for
\(0<\alpha<1\). Hence, using \eqref{eq:AD-Z-invariance},
\begin{align}
    Q_{\alpha}(\Delta_Z(\rho);\mathscr N)
    &\leq
    \frac12 Q_{\alpha}(\rho;\mathscr N)
    +
    \frac12 Q_{\alpha}(Z\rho Z;\mathscr N)
    \notag\\
    &=
    Q_{\alpha}(\rho;\mathscr N).
\end{align}
Thus every state can be replaced by a diagonal state without increasing
the objective. Since diagonal states are themselves admissible, this
proves \eqref{eq:AD-diagonal-reduction}. The minimum is attained by
finite-dimensional continuity and compactness of the state space.
\end{proof}

Set
\begin{align}
    u:=1-t,
    \qquad
    \ell(t):=t+\eta u=\eta+\gamma t.
    \label{eq:AD-u-ell}
\end{align}
For the diagonal state \(\rho_t\), define
\begin{align}
    |\psi_t\rangle
    :=
    \sqrt t\,|00\rangle+\sqrt{\eta u}\,|11\rangle.
\end{align}
Using \eqref{eq:AD-Choi-branches}, we have
\begin{align}
    \sqrt{\rho_t}\,
    \Gamma^{\mathscr N}\,
    \sqrt{\rho_t}
    =
    |\psi_t\rangle\!\langle\psi_t|
    +
    \gamma u\,|10\rangle\!\langle10|.
    \label{eq:AD-one-copy-blocks}
\end{align}
The two summands have orthogonal supports, and
\(\langle\psi_t|\psi_t\rangle=\ell(t)\). Therefore
\begin{align}
    \left(
        \sqrt{\rho_t}\,
        \Gamma^{\mathscr N}\,
        \sqrt{\rho_t}
    \right)^\alpha
    =
    \ell(t)^{\alpha-1}
    |\psi_t\rangle\!\langle\psi_t|
    +
    (\gamma u)^\alpha
    |10\rangle\!\langle10|.
\end{align}
After multiplying by \(\rho_t^{1-\alpha}\) and tracing out \(\A\),
we obtain
\begin{align}
    T_{\alpha}^{\mathscr N}(\rho_t)
    =
    w_0(t)|0\rangle\!\langle0|
    +
    w_1(t)|1\rangle\!\langle1|,
    \label{eq:AD-one-copy-T}
\end{align}
where
\begin{align}
    w_0(t)
    &=
    t^{2-\alpha}\ell(t)^{\alpha-1}
    +
    \gamma^\alpha(1-t),
    \label{eq:w0-direct}
    \\
    w_1(t)
    &=
    \eta(1-t)^{2-\alpha}\ell(t)^{\alpha-1}.
    \label{eq:w1-direct}
\end{align}
Consequently,
\begin{align}
    q_{\alpha}(t)
    &:=
    Q_{\alpha}(\rho_t;\mathscr N)
    =
    w_0(t)^{1/\alpha}
    +
    w_1(t)^{1/\alpha}.
    \label{eq:AD-one-copy-q}
\end{align}

\begin{app-lemma}[Convexity and interiority of the one-copy minimizer]
\label{lemm:AD-reduced-convexity}
For every \(0<\gamma<1\) and \(0<\alpha<1\), the function
\(q_{\alpha}\) is convex on \([0,1]\). Moreover,
\begin{align}
    q_{\alpha}'(0+)<0,
    \qquad
    q_{\alpha}'(1-)>0.
\end{align}
In particular, every minimizer \(t_\alpha\) of \(q_\alpha\) belongs to
\((0,1)\).
\end{app-lemma}

\begin{proof}
Since \(2-\alpha>1\), the function \(r\mapsto r^{2-\alpha}\) is
convex on \([0,\infty)\). Its perspective
\begin{align}
    (r,z)
    \mapsto
    z\left(\frac rz\right)^{2-\alpha}
    =
    r^{2-\alpha}z^{\alpha-1}
\end{align}
is therefore jointly convex for \(r\geq0\) and \(z>0\).
Because \(t\), \(1-t\), and \(\ell(t)\) are affine in \(t\), both
\begin{align}
    t^{2-\alpha}\ell(t)^{\alpha-1},
    \qquad
    (1-t)^{2-\alpha}\ell(t)^{\alpha-1}
\end{align}
are convex. Thus \(w_0\) and \(w_1\) are nonnegative convex functions.
Since \(r\mapsto r^{1/\alpha}\) is convex and increasing for
\(0<\alpha<1\), the function
\(q_\alpha=w_0^{1/\alpha}+w_1^{1/\alpha}\) is convex.

A direct evaluation of the one-sided derivatives gives
\begin{align}
    q_{\alpha}'(0+)
    =
    -\frac{2-\alpha}{\alpha}
    <0,
    \label{eq:AD-q-derivative-zero}
\end{align}
and
\begin{align}
    q_{\alpha}'(1-)
    =
    \frac{
        2-\alpha+(\alpha-1)\gamma-\gamma^\alpha
    }{\alpha}.
    \label{eq:AD-q-derivative-one}
\end{align}
To see that the latter is positive, define
\begin{align}
    h_\alpha(\gamma)
    :=
    2-\alpha+(\alpha-1)\gamma-\gamma^\alpha.
\end{align}
Then
\begin{align}
    h_\alpha(1)=0,
    \qquad
    h_\alpha'(\gamma)
    =
    \alpha-1-\alpha\gamma^{\alpha-1}
    <0.
\end{align}
Hence \(h_\alpha(\gamma)>h_\alpha(1)=0\) whenever \(0<\gamma<1\).
Neither endpoint can therefore minimize \(q_\alpha\), proving the
claim.
\end{proof}

Fix any minimizer \(t_\alpha\in(0,1)\), and abbreviate
\begin{align}
    t:=t_\alpha,
    \qquad
    u:=1-t,
    \qquad
    \ell:=\eta+\gamma t,
    \qquad
    w_i:=w_i(t),\quad i\in\{0,1\}.
    \label{eq:AD-minimizer-abbreviations}
\end{align}
Then
\begin{align}
    K^{(1)}(\alpha)
    =
    q_\alpha(t)
    =
    w_0^{1/\alpha}+w_1^{1/\alpha}.
    \label{eq:AD-K1}
\end{align}

\subsection{A correlated diagonal two-copy path}

Consider the diagonal two-copy state
\begin{align}
    \rho_\alpha(\kappa)
    &=
    \rho_t\otimes\rho_t
    +
    \kappa\,
    \mathrm{diag}(1,-1,-1,1)
    \notag\\
    &=
    \mathrm{diag}\left(
        p_{00}(\kappa),
        p_{01}(\kappa),
        p_{10}(\kappa),
        p_{11}(\kappa)
    \right),
    \label{eq:two-copy_GAD}
\end{align}
where
\begin{align}
    p_{00}(\kappa)&=t^2+\kappa,
    &
    p_{01}(\kappa)&=tu-\kappa,
    \notag\\
    p_{10}(\kappa)&=tu-\kappa,
    &
    p_{11}(\kappa)&=u^2+\kappa.
    \label{eq:AD-p-kappa}
\end{align}
This is a full-rank density operator whenever
\begin{align}
    -\min\{t^2,u^2\}<\kappa<tu.
    \label{eq:AD-kappa-interval}
\end{align}
Moreover, the perturbation preserves both one-copy marginals:
\begin{align}
    \Tr_{\A_2}\rho_\alpha(\kappa)
    =
    \Tr_{\A_1}\rho_\alpha(\kappa)
    =
    \rho_t.
    \label{eq:AD-fixed-marginals}
\end{align}
Thus positive \(\kappa\) introduces a classical correlation without
changing either marginal.

For product states and product channels, the functional is
multiplicative:
\begin{align}
    Q_\alpha(\rho\otimes\sigma;\mathscr N\otimes\mathscr M)
    =
    Q_\alpha(\rho;\mathscr N)
    Q_\alpha(\sigma;\mathscr M).
    \label{eq:AD-Q-product}
\end{align}
Indeed, the operator \(T_\alpha\) in \eqref{eq:AD-T-def} factors as a
tensor product, and both the \(1/\alpha\)-power and the trace factor.
Consequently,
\begin{align}
    Q_\alpha(\rho_\alpha(0);\mathscr N^{\otimes2})
    =
    Q_\alpha(\rho_t;\mathscr N)^2
    =
    \bigl(K^{(1)}(\alpha)\bigr)^2.
    \label{eq:AD-product-value}
\end{align}

\subsection{The four two-copy output weights}

The decomposition \eqref{eq:AD-Choi-branches} gives four two-copy
branches
\begin{align}
    |\phi_r\rangle_{\A_1\B_1}
    \otimes
    |\phi_s\rangle_{\A_2\B_2},
    \qquad
    r,s\in\{0,1\}.
\end{align}
For the diagonal state \(\rho_\alpha(\kappa)\), define
\begin{align}
    |\psi_{rs}(\kappa)\rangle
    :=
    \left(
        \sqrt{\rho_\alpha(\kappa)_{\A_1\A_2}}
        \otimes\I_{\B_1\B_2}
    \right)
    |\phi_r\rangle_{\A_1\B_1}
    |\phi_s\rangle_{\A_2\B_2}.
\end{align}
The four vectors \( |\psi_{rs}(\kappa)\rangle \) have mutually
orthogonal supports. Hence
\begin{align}
    \sqrt{\rho_\alpha(\kappa)}\,
    (\Gamma^{\mathscr N})^{\otimes2}\,
    \sqrt{\rho_\alpha(\kappa)}
    =
    \sum_{r,s=0}^1
    |\psi_{rs}(\kappa)\rangle
    \!\langle\psi_{rs}(\kappa)|
\end{align}
is an orthogonal sum of rank-one operators.

Let
\begin{align}
    \lambda_{rs}(\kappa)
    :=
    \langle\psi_{rs}(\kappa)|\psi_{rs}(\kappa)\rangle.
\end{align}
Explicitly,
\begin{align}
    \lambda_{00}(\kappa)
    &=
    p_{00}
    +
    \eta(p_{01}+p_{10})
    +
    \eta^2p_{11}
    =
    \ell^2+\gamma^2\kappa,
    \label{eq:AD-lambda00}
    \\
    \lambda_{01}(\kappa)
    &=
    \gamma(p_{01}+\eta p_{11})
    =
    \gamma(u\ell-\gamma\kappa),
    \label{eq:AD-lambda01}
    \\
    \lambda_{10}(\kappa)
    &=
    \gamma(p_{10}+\eta p_{11})
    =
    \gamma(u\ell-\gamma\kappa),
    \label{eq:AD-lambda10}
    \\
    \lambda_{11}(\kappa)
    &=
    \gamma^2p_{11}
    =
    \gamma^2(u^2+\kappa).
    \label{eq:AD-lambda11}
\end{align}
Here and below, the argument \(\kappa\) of the \(p_{ij}\)'s is
suppressed.

Since
\begin{align}
    \bigl(
        |\psi_{rs}\rangle\!\langle\psi_{rs}|
    \bigr)^\alpha
    =
    \lambda_{rs}^{\alpha-1}
    |\psi_{rs}\rangle\!\langle\psi_{rs}|,
\end{align}
the operator after tracing out the two input systems is diagonal:
\begin{align}
    T_\alpha^{\mathscr N^{\otimes2}}
    \bigl(\rho_\alpha(\kappa)\bigr)
    =
    \sum_{i,j=0}^1
    w_{ij}(\kappa)|ij\rangle\!\langle ij|.
    \label{eq:AD-two-copy-T}
\end{align}
The four diagonal entries are
\begin{align}
    w_{00}(\kappa)
    &=
    p_{00}^{2-\alpha}\lambda_{00}^{\alpha-1}
    +
    \gamma p_{01}^{2-\alpha}\lambda_{01}^{\alpha-1}
    +
    \gamma p_{10}^{2-\alpha}\lambda_{10}^{\alpha-1}
    +
    \gamma^2p_{11}^{2-\alpha}\lambda_{11}^{\alpha-1},
    \label{eq:AD-w00}
    \\
    w_{01}(\kappa)
    &=
    \eta p_{01}^{2-\alpha}\lambda_{00}^{\alpha-1}
    +
    \gamma\eta p_{11}^{2-\alpha}\lambda_{10}^{\alpha-1},
    \label{eq:AD-w01}
    \\
    w_{10}(\kappa)
    &=
    \eta p_{10}^{2-\alpha}\lambda_{00}^{\alpha-1}
    +
    \gamma\eta p_{11}^{2-\alpha}\lambda_{01}^{\alpha-1},
    \label{eq:AD-w10}
    \\
    w_{11}(\kappa)
    &=
    \eta^2p_{11}^{2-\alpha}\lambda_{00}^{\alpha-1}.
    \label{eq:AD-w11}
\end{align}
For example, the four terms in \(w_{00}\) correspond, respectively,
to no jumps, a jump in the second use, a jump in the first use, and
jumps in both uses. The two terms in \(w_{01}\) correspond to the
input \(01\) with no jump and the input \(11\) with a jump in the
first use.

It follows that
\begin{align}
    Q_\alpha(\rho_\alpha(\kappa);\mathscr N^{\otimes2})
    =
    \sum_{i,j=0}^1 w_{ij}(\kappa)^{1/\alpha}.
    \label{eq:AD-Q-two-copy-weights}
\end{align}
At \(\kappa=0\),
\begin{align}
    \lambda_{00}(0)&=\ell^2,
    &
    \lambda_{01}(0)&=\gamma u\ell,
    \notag\\
    \lambda_{10}(0)&=\gamma u\ell,
    &
    \lambda_{11}(0)&=\gamma^2u^2,
\end{align}
and the product structure gives
\begin{align}
    w_{ij}(0)=w_iw_j,
    \qquad i,j\in\{0,1\}.
    \label{eq:AD-w-product}
\end{align}

\subsection{Factorized linear response}

Define
\begin{align}
    a_0
    &:=
    \gamma t^{2-\alpha}\ell^{\alpha-2}
    -
    \gamma^\alpha,
    &
    a_1
    &:=
    \gamma\eta u^{2-\alpha}\ell^{\alpha-2}
    =
    \frac{\gamma w_1}{\ell},
    \label{eq:AD-a-def}
    \\
    b_0
    &:=
    t^{1-\alpha}\ell^{\alpha-1}
    -
    \gamma^\alpha,
    &
    b_1
    &:=
    -\eta u^{1-\alpha}\ell^{\alpha-1}
    =
    -\frac{w_1}{u}.
    \label{eq:AD-b-def}
\end{align}

\begin{app-lemma}[Factorized first variation]
\label{lemm:AD-factorized-response}
The one-copy derivatives satisfy
\begin{align}
    w_i'(t)
    =
    (\alpha-1)a_i+(2-\alpha)b_i,
    \qquad i\in\{0,1\}.
    \label{eq:split-direct}
\end{align}
Moreover, along the correlated path
\eqref{eq:two-copy_GAD},
\begin{align}
    \left.
    \frac{\mathrm d}{\mathrm d\kappa}
    w_{ij}(\kappa)
    \right|_{\kappa=0}
    =
    (\alpha-1)a_i a_j
    +
    (2-\alpha)b_i b_j,
    \qquad
    i,j\in\{0,1\}.
    \label{eq:AD-factorized-wdot}
\end{align}
Equivalently,
\begin{align}
    \left.
    \frac{\mathrm d}{\mathrm d\kappa}
    \begin{pmatrix}
        w_{00}(\kappa)&w_{01}(\kappa)\\
        w_{10}(\kappa)&w_{11}(\kappa)
    \end{pmatrix}
    \right|_{\kappa=0}
    =
    (\alpha-1)
    \begin{pmatrix}a_0\\a_1\end{pmatrix}
    \begin{pmatrix}a_0&a_1\end{pmatrix}
    +
    (2-\alpha)
    \begin{pmatrix}b_0\\b_1\end{pmatrix}
    \begin{pmatrix}b_0&b_1\end{pmatrix}.
    \label{eq:wdot-direct}
\end{align}
\end{app-lemma}

\begin{proof}
We give a factorized derivation of the identity.

Let
\begin{align}
    r_0:=t,
    \qquad
    r_1:=u,
    \qquad
    d_0:=1,
    \qquad
    d_1:=-1.
\end{align}
Thus
\begin{align}
    \frac{\mathrm d}{\mathrm dt}r_x=d_x,
    \qquad
    p_{xy}(\kappa)=r_xr_y+\kappa d_xd_y.
    \label{eq:AD-correlation-factorization}
\end{align}

For a single channel use, let \(s=0\) denote the no-jump branch and
\(s=1\) the jump branch. Denote by \(c_s(x)\) the squared branch
amplitude for input \(x\):
\begin{align}
    c_0(0)=1,
    \qquad
    c_0(1)=\eta,
    \qquad
    c_1(0)=0,
    \qquad
    c_1(1)=\gamma.
\end{align}
The corresponding branch eigenvalues and their variations in the
direction \(d=(1,-1)\) are
\begin{align}
    \lambda_s
    &:=
    \sum_{x=0}^1 r_xc_s(x),
    &
    \delta\lambda_s
    &:=
    \sum_{x=0}^1 d_xc_s(x).
\end{align}
Thus
\begin{align}
    \lambda_0=\ell,
    \qquad
    \lambda_1=\gamma u,
    \qquad
    \delta\lambda_0=\gamma,
    \qquad
    \delta\lambda_1=-\gamma.
    \label{eq:AD-single-branch-data}
\end{align}

The input-branch pairs producing output \(i\) are
\begin{align}
    \mathcal E_0
    =
    \{(0,0),(1,1)\},
    \qquad
    \mathcal E_1
    =
    \{(1,0)\}.
\end{align}
Accordingly,
\begin{align}
    w_i
    =
    \sum_{(x,s)\in\mathcal E_i}
    r_x^{2-\alpha}
    c_s(x)
    \lambda_s^{\alpha-1}.
    \label{eq:AD-w-branch-form}
\end{align}
Differentiating \eqref{eq:AD-w-branch-form} with respect to \(t\)
gives
\begin{align}
    w_i'
    &=
    (\alpha-1)
    \sum_{(x,s)\in\mathcal E_i}
    r_x^{2-\alpha}
    c_s(x)
    \lambda_s^{\alpha-2}
    \delta\lambda_s
    \notag\\
    &\quad+
    (2-\alpha)
    \sum_{(x,s)\in\mathcal E_i}
    d_xr_x^{1-\alpha}
    c_s(x)
    \lambda_s^{\alpha-1}.
    \label{eq:AD-one-copy-factorized-derivative}
\end{align}
The first sum equals \(a_i\), and the second equals \(b_i\).
Evaluating the three possible pairs
\((0,0),(1,1),(1,0)\) gives exactly
\eqref{eq:AD-a-def}--\eqref{eq:AD-b-def}, proving
\eqref{eq:split-direct}.

For two channel uses, define the branch-pair eigenvalue
\begin{align}
    \lambda_{s_1s_2}^{(2)}(\kappa)
    :=
    \sum_{x,y=0}^1
    p_{xy}(\kappa)c_{s_1}(x)c_{s_2}(y).
\end{align}
Using \eqref{eq:AD-correlation-factorization}, we find
\begin{align}
    \lambda_{s_1s_2}^{(2)}(0)
    &=
    \lambda_{s_1}\lambda_{s_2},
    \label{eq:AD-lambda-product}
    \\
    \left.
    \frac{\mathrm d}{\mathrm d\kappa}
    \lambda_{s_1s_2}^{(2)}(\kappa)
    \right|_{\kappa=0}
    &=
    \delta\lambda_{s_1}\delta\lambda_{s_2}.
    \label{eq:AD-lambda-response-product}
\end{align}
The two-copy output weight has the branch representation
\begin{align}
    w_{ij}(\kappa)
    =
    \sum_{\substack{(x,s_1)\in\mathcal E_i\\
                    (y,s_2)\in\mathcal E_j}}
    p_{xy}(\kappa)^{2-\alpha}
    c_{s_1}(x)c_{s_2}(y)
    \left(
        \lambda_{s_1s_2}^{(2)}(\kappa)
    \right)^{\alpha-1}.
    \label{eq:AD-wij-branch-form}
\end{align}
Differentiating at \(\kappa=0\), the variation of the branch
eigenvalue contributes
\begin{align}
    &(\alpha-1)
    \sum_{\substack{(x,s_1)\in\mathcal E_i\\
                    (y,s_2)\in\mathcal E_j}}
    (r_xr_y)^{2-\alpha}
    c_{s_1}(x)c_{s_2}(y)
    (\lambda_{s_1}\lambda_{s_2})^{\alpha-2}
    \delta\lambda_{s_1}\delta\lambda_{s_2}
    \notag\\
    &\qquad=
    (\alpha-1)a_i a_j.
\end{align}
The variation of the input probabilities contributes
\begin{align}
    &(2-\alpha)
    \sum_{\substack{(x,s_1)\in\mathcal E_i\\
                    (y,s_2)\in\mathcal E_j}}
    d_xd_y(r_xr_y)^{1-\alpha}
    c_{s_1}(x)c_{s_2}(y)
    (\lambda_{s_1}\lambda_{s_2})^{\alpha-1}
    \notag\\
    &\qquad=
    (2-\alpha)b_i b_j.
\end{align}
This proves \eqref{eq:AD-factorized-wdot}.
\end{proof}

Define
\begin{align}
    x_\alpha
    &:=
    w_0^{1/\alpha-1}a_0
    +
    w_1^{1/\alpha-1}a_1,
    \label{eq:AD-x-def}
    \\
    y_\alpha
    &:=
    w_0^{1/\alpha-1}b_0
    +
    w_1^{1/\alpha-1}b_1.
    \label{eq:AD-y-def}
\end{align}
Since \(t=t_\alpha\) is an interior minimizer, stationarity and
\eqref{eq:split-direct} give
\begin{align}
    0
    =
    q_\alpha'(t)
    =
    \frac1\alpha
    \left[
        (\alpha-1)x_\alpha
        +
        (2-\alpha)y_\alpha
    \right].
\end{align}
Thus
\begin{align}
    (\alpha-1)x_\alpha
    +
    (2-\alpha)y_\alpha
    =
    0.
    \label{eq:xy-stationarity-direct}
\end{align}

Using \eqref{eq:AD-w-product} and
Lemma~\ref{lemm:AD-factorized-response}, we obtain
\begin{align}
    &\left.
    \frac{\mathrm d}{\mathrm d\kappa}
    Q_\alpha(\rho_\alpha(\kappa);\mathscr N^{\otimes2})
    \right|_{\kappa=0}
    \notag\\
    &\quad=
    \frac1\alpha
    \sum_{i,j=0}^1
    (w_iw_j)^{1/\alpha-1}
    \left[
        (\alpha-1)a_ia_j
        +
        (2-\alpha)b_ib_j
    \right]
    \notag\\
    &\quad=
    \frac1\alpha
    \left[
        (\alpha-1)x_\alpha^2
        +
        (2-\alpha)y_\alpha^2
    \right].
    \label{eq:response-pre-direct}
\end{align}
Eliminating \(y_\alpha\) by
\eqref{eq:xy-stationarity-direct} gives
\begin{align}
    \left.
    \frac{\mathrm d}{\mathrm d\kappa}
    Q_\alpha(\rho_\alpha(\kappa);\mathscr N^{\otimes2})
    \right|_{\kappa=0}
    =
    \frac{\alpha-1}{\alpha(2-\alpha)}
    x_\alpha^2.
    \label{eq:response-x-direct}
\end{align}

It remains to prove that \(x_\alpha\neq0\). Suppose, to the contrary,
that \(x_\alpha=0\). Since \(2-\alpha>0\),
\eqref{eq:xy-stationarity-direct} then gives \(y_\alpha=0\).
Therefore the strictly positive vector
\begin{align}
    z
    :=
    \left(
        w_0^{1/\alpha-1},
        w_1^{1/\alpha-1}
    \right)
\end{align}
is orthogonal to both
\((a_0,a_1)\) and \((b_0,b_1)\). In two dimensions this would imply
that these latter two vectors are linearly dependent.

On the other hand, direct simplification gives
\begin{align}
    a_0b_1-a_1b_0
    &=
    \eta\gamma^\alpha
    u^{1-\alpha}\ell^{\alpha-2}
    -
    \eta\gamma
    t^{1-\alpha}u^{1-\alpha}
    \ell^{2\alpha-3}
    \notag\\
    &=
    \eta\gamma\,
    t^{1-\alpha}u^{1-\alpha}
    \ell^{\alpha-2}
    \left[
        (\gamma t)^{\alpha-1}
        -
        \ell^{\alpha-1}
    \right].
    \label{eq:det-direct}
\end{align}
All factors outside the brackets are strictly positive. Moreover,
\begin{align}
    0<\gamma t<\ell
\end{align}
and \(\alpha-1<0\), so
\begin{align}
    (\gamma t)^{\alpha-1}
    >
    \ell^{\alpha-1}.
\end{align}
Hence
\begin{align}
    a_0b_1-a_1b_0>0,
\end{align}
contradicting linear dependence. Therefore \(x_\alpha\neq0\).

Since \(0<\alpha<1\), we conclude from
\eqref{eq:response-x-direct} that
\begin{align}
    \left.
    \frac{\mathrm d}{\mathrm d\kappa}
    Q_\alpha(\rho_\alpha(\kappa);\mathscr N^{\otimes2})
    \right|_{\kappa=0}
    <0.
    \label{eq:AD-negative-response}
\end{align}

\begin{app-theorem}[Strict superadditivity of amplitude damping]
\label{thm:AD-strict-superadditivity}
Let \(\mathscr N\) be the amplitude-damping channel
\eqref{eq:GAD-app} with \(0<\gamma<1\). Then
\begin{align}
    I_\alpha(\mathscr N^{\otimes2})
    >
    2I_\alpha(\mathscr N)
    \qquad
    \forall\,0<\alpha<1.
\end{align}
The strict improvement is witnessed by the diagonal,
classically correlated path \(\rho_\alpha(\kappa)\) in
\eqref{eq:two-copy_GAD}.
\end{app-theorem}

\begin{proof}
By \eqref{eq:AD-negative-response}, for every sufficiently small
\(\kappa>0\),
\begin{align}
    Q_\alpha(\rho_\alpha(\kappa);\mathscr N^{\otimes2})
    <
    Q_\alpha(\rho_\alpha(0);\mathscr N^{\otimes2})
    =
    \bigl(K^{(1)}(\alpha)\bigr)^2.
\end{align}
Therefore
\begin{align}
    K^{(2)}(\alpha)
    &\leq
    Q_\alpha(\rho_\alpha(\kappa);\mathscr N^{\otimes2})
    <
    \bigl(K^{(1)}(\alpha)\bigr)^2.
\end{align}
Finally, because \(\alpha/(\alpha-1)<0\),
\begin{align}
    I_\alpha(\mathscr N^{\otimes2})
    &=
    \frac{\alpha}{\alpha-1}\log K^{(2)}(\alpha)
    \notag\\
    &>
    \frac{\alpha}{\alpha-1}
    \log\bigl(K^{(1)}(\alpha)^2\bigr)
    =
    2I_\alpha(\mathscr N).
\end{align}
\end{proof}


\section{Single-Letter Bounds} \label{sec:single-letter}

Except for the special channels (Proposition~\ref{prop:additivity_special-app}) and $\alpha = 1,2$, the Petz--\Renyi information $I_{\alpha}(\mathscr{N})$ can be strictly superadditive as shown in Sections~\ref{sec:q-c} and \ref{sec:GAD}.
This means that the achievable random coding exponent $E_{\mathrm{r}}(R;\mathscr{N})$ has a multi-letter expression in general.

Define the regularized Petz channel information as
\begin{align}
	I_{\alpha}^{\infty}(\mathscr{N})
	&\coloneq \sup_{n\in\mathds{N}} \frac{1}{n} I_{\alpha}\left(\mathscr{N}^{\otimes n}\right)
	= \lim_{n\to\infty} \frac{1}{n} I_{\alpha}\left(\mathscr{N}^{\otimes n}\right),
\end{align}
where the last equality follows from the superadditive sequence and Fekete’s lemma.

\begin{app-lemma}[\hspace{-0.5pt}{\cite[Proposition 11]{Wil18}}, {\cite[Corollary 3.6]{Jen2018_I}}, {\cite[Lemma 5.4]{ATB24}}] \label{lemm:Petz_sandwiched}
	Let $\rho$ and $\sigma$ be density operators.
	Then,
	\begin{align}
		D_{2-\frac{1}{\alpha}}(\rho\Vert\sigma)
		\leq 
		\widetilde{D}_{\alpha}(\rho\Vert\sigma), \quad \forall\, \alpha \geq \tfrac{1}{2}.
	\end{align}
	
	Equivalently,
	\begin{align}
		D_{\alpha}(\rho\Vert\sigma) \leq \widetilde{D}_{\frac{1}{2-\alpha}}(\rho\Vert\sigma), \quad \forall \alpha \in [0,2].
	\end{align}
\end{app-lemma}

\begin{app-proposition}[Single-letter bounds] \label{prop:single-letter}
	Let $\mathscr{N}_{\A\to\B}$ be a quantum channel.
	We have
	\begin{align}
		I_{\alpha}(\mathscr{N})
		\leq 
		I_{\alpha}^{\infty}(\mathscr{N})
		\leq
		\widetilde{I}_{\frac{1}{2-\alpha}}(\mathscr{N})
		\leq 
		{I}_{\frac{1}{2-\alpha}}(\mathscr{N}),
		\quad \forall\,\alpha \in [0,2].
	\end{align}
\end{app-proposition}
\begin{proof}
For any integer $n$ and $\alpha \in [0,2]$, 
Lemma~\ref{lemm:Petz_sandwiched} implies that
\begin{align}
		I_{\alpha}(\mathscr{N})
		\leq 
		\frac{1}{n}I_{\alpha}(\mathscr{N}^{\otimes n})
		\leq
		\frac{1}{n}\widetilde{I}_{\frac{1}{2-\alpha}}(\mathscr{N}^{\otimes n})
		= \widetilde{I}_{\frac{1}{2-\alpha}}(\mathscr{N}),
	\end{align}
    where the last equality follows from the additivity of $\widetilde{I}_{\beta}(\mathscr{N})$, $\beta = \frac{1}{2-\alpha}\in[\sfrac12,1)$ \cite{li2026completely} and $\beta = \frac{1}{2-\alpha}\geq 1$ \cite{GW14}.
Taking $n\to\infty$ completes the proof.
\end{proof}

The single-letter upper bound in terms of  $\widetilde{I}_{\frac{1}{2-\alpha}}(\mathscr{N})$ is a double-state optimization; see~\eqref{eq:defn_sandwiched-app}.
One may further relax it to a more computationally feasible one-state optimization of the Petz form ${I}_{\frac{1}{2-\alpha}}(\mathscr{N})$.
Below, we show that the sandwiched \Renyi information $\widetilde{I}_{\alpha}(\mathscr{N})$ admits a one-state optimization expression for quantum-classical channels.

\begin{app-proposition}[One-state optimization for quantum-classical channels]
	\label{proposition:sandwiched_sibson}
	Let
	\(\alpha\in[\sfrac12,1)\cup(1,+\infty)\)
	and let \(\mathscr{M}=\{M_y\}_{y\in\Y}\) be a finite POVM. Define
	\begin{equation}
		\widetilde{Q}_y(\rho):=
		\Tr\left[
			\left(
				\rho^{1/(2\alpha)}
				M_y^{\top}
				\rho^{1/(2\alpha)}
			\right)^\alpha
		\right].
	\end{equation}
	Then
	\begin{equation}
		\widetilde{I}_\alpha(\mathscr{M})
		=
		\begin{cases}
			\displaystyle
			\frac{\alpha}{\alpha-1}
			\log
			\inf_{\rho}
			\sum_{y\in\Y}
			\widetilde{Q}_y(\rho)^{1/\alpha},
			&
			\alpha\in[\sfrac12,1),
			\\[10pt]
			\displaystyle
			\frac{\alpha}{\alpha-1}
			\log
			\sup_{\rho}
			\sum_{y\in\Y}
			\widetilde{Q}_y(\rho)^{1/\alpha},
			&
			\alpha\in(1,+\infty).
		\end{cases}
		\label{eq:K-reduction}
	\end{equation}
	For fixed \(\rho\), an optimizing auxiliary output distribution is
\begin{equation}
	q_y^\star(\rho)
	=
	\frac{\widetilde{Q}_y(\rho)^{1/\alpha}}
	{\sum_{\bar{y}\in\Y}
		\widetilde{Q}_{\bar{y}}(\rho)^{1/\alpha}},
	\label{eq:sigma-opt}
\end{equation}
with the convention that
\(q_y^\star(\rho)=0\) whenever
\(\widetilde{Q}_y(\rho)=0\).
\end{app-proposition}

\begin{proof}
	Write
	\(\sigma_{\Y}=\sum_y q_y|y\rangle\langle y|_{\Y}\).
	Since the output is classical, we have
	\begin{align}
		&\Tr\left[
			\left(
				(\rho\otimes\sigma_{\Y})^{\frac{1-\alpha}{2\alpha}}
				\omega_{\R\Y}^{\rho}
				(\rho\otimes\sigma_{\Y})^{\frac{1-\alpha}{2\alpha}}
			\right)^\alpha
		\right]
		\\
		&\qquad
		=
		\sum_y q_y^{1-\alpha}
		\Tr\left[
			\left(
				\rho^{\frac{1-\alpha}{2\alpha}}
				\rho^{1/2}M_y^{\top}\rho^{1/2}
				\rho^{\frac{1-\alpha}{2\alpha}}
			\right)^\alpha
		\right]
		\\
		&\qquad
		=
		\sum_y q_y^{1-\alpha}\widetilde{Q}_y(\rho).
	\end{align}
	For \(\alpha\in[\sfrac12,1)\), since
	\(1/(\alpha-1)<0\), minimizing the divergence over \(q\)
	is equivalent to maximizing
	\(\sum_yq_y^{1-\alpha}\widetilde{Q}_y(\rho)\)
	over the probability simplex.
	For \(\alpha>1\), since \(1/(\alpha-1)>0\), minimizing the
	divergence is equivalent to minimizing the same expression.
	In the latter case, the support condition requires
	\(q_y>0\) whenever \(\widetilde{Q}_y(\rho)>0\); terms for which
	\(\widetilde{Q}_y(\rho)=0\) may be omitted.

	The scalar objective is concave in \(q\) for \(\alpha<1\) and
	convex in \(q\) for \(\alpha>1\). In both cases, the optimality condition gives
	\begin{equation}
		q_y\propto\widetilde{Q}_y(\rho)^{1/\alpha}
	\end{equation}
    for $y:\widetilde{Q}_y(\rho)>0$
	and the optimal value is
	\begin{equation}
		\left(
			\sum_y\widetilde{Q}_y(\rho)^{1/\alpha}
		\right)^\alpha.
	\end{equation}
	Hence, for either
	\(\alpha\in[\sfrac12,1)\) or \(\alpha>1\),
	\begin{equation}
		\inf_{\sigma_{\Y}}
		\widetilde{D}_\alpha
		\left(
			\omega_{\R\Y}^{\rho}
			\middle\|
			\rho\otimes\sigma_{\Y}
		\right)
		=
		\frac{\alpha}{\alpha-1}
		\log
		\sum_y\widetilde{Q}_y(\rho)^{1/\alpha}.
	\end{equation}
	For \(\alpha<1\), the coefficient
	\(\alpha/(\alpha-1)\) is negative, so the outer supremum over
	\(\rho\) becomes the infimum in \eqref{eq:K-reduction}.
	For \(\alpha>1\), this coefficient is positive, so the outer
	supremum remains the supremum in \eqref{eq:K-reduction}.
\end{proof}

\end{document}